\pdfoutput=1 
\documentclass[12pt]{article}
\usepackage{amsmath,amsthm,amssymb}
\usepackage{multirow}
\usepackage{times}
\usepackage{enumerate}
\usepackage[margin=1.25 in]{geometry} 
\usepackage{dsfont}
\usepackage{float}
\usepackage[title]{appendix}

\usepackage{natbib}

\usepackage{setspace}
\usepackage[colorlinks,
linkcolor=blue,
anchorcolor=blue,
citecolor=blue
]{hyperref}

\theoremstyle{plain}

\newtheorem{proposition}{Proposition}
\newtheorem{theorem}{Theorem}
\newtheorem{corollary}{Corollary}

\newtheorem{lemma}[theorem]{Lemma}

\newtheorem{definition}{Definition}
\newtheorem*{example}{Example}
\newtheorem{remark}{Remark}
\newtheorem{assumption}{Assumption}

\newcommand{\N}{\mathbb{N}}

\newcommand{\R}{\mathbb{R}}

\makeatletter

\newcommand{\Rmnum}[1]{\expandafter\@slowromancap\romannumeral #1@}

\title{Improving Robust Decisions with Data\footnote{This paper is a revised version of a chapter in my Northwestern University doctoral dissertation. The extended abstract of an earlier version appeared in EC'22. I am deeply indebted to Peter Klibanoff and Marciano Siniscalchi for their guidance and support. Comments from the editor and four anonymous referees have greatly improved the paper. I also thank Federico Bugni, Ivan Canay, Tristan Gagnon-Bartsch, Shaowei Ke, Yonggyun Kim, R.Vijay Krishna, Yijun Liu, Edwin Munoz Rodriguez, Isaias N.Chaves, Eran Shmaya, and Chen Zhao for insightful comments.}}

\author{Xiaoyu Cheng\footnote{Department of Economics, Florida State University, Tallahassee, FL, USA. E-mail: \href{mailto:xcheng@fsu.edu}{xcheng@fsu.edu}.}}

\begin{document}
\maketitle

\begin{abstract}
    A decision-maker faces uncertainty governed by a data-generating process (DGP), which is only known to belong to a set of sequences of independent but possibly non-identical distributions. A robust decision maximizes the expected payoff against the worst possible DGP in this set. This paper characterizes when and how such robust decisions can be \emph{objectively} improved with data --- that is, yield higher expected payoffs under the true DGP regardless of which DGP is the truth. It further develops simple and novel inference procedures that achieve such improvement, while common methods (e.g., maximum likelihood) may fail to do so.


\end{abstract}

\newpage 
\section{Introduction}\label{intro}
When a decision-maker (DM) lacks sufficient knowledge about the probability law governing an uncertain environment, they may want their decisions to be robust. This concern for robustness is often modeled by ranking choices by their worst-case expected payoffs among a set of possible laws \citep{10.1214/aoms/1177698602, GILBOA1989141, HansenSargent+2007, Carroll2019}. With insufficient knowledge, gathering data is a classic means of learning about the environment to help guide decisions. Here, learning involves using data to revise the set of possible laws. A revision with data is said to provide an \textit{objective improvement} if the data-revised robust decision yields a higher expected payoff under the \textit{true} law than the robust decision without any data-revision, the benchmark robust decision. This paper studies when and how data may be used to achieve such an objective improvement regardless of which possible law is the true one. 

The analysis is conducted in a decision environment where it is a priori unclear whether and how data can generate objective improvement --- one in which the true law cannot be uniquely identified even asymptotically. Specifically, the DM faces a countable sequence of random experiments (e.g., coin flips) that share the same set of outcomes. The probability law governing these realizations, referred to as the data-generating process (DGP), is a sequence of independent but possibly non-identical distributions over the outcomes. The DM initially knows there is a set of possible DGPs that contains the true one, observes sample data given by realizations of some experiments, and then makes a decision whose payoff depends only on future outcomes. 

This non-identical decision environment captures settings with unobserved heterogeneity. For example, consider an online platform, say Netflix, making content recommendations to users based on feedback from other users with similar profiles. Users' preferences, however, are usually also determined by their various offline activities that cannot be observed, which often further affect how they interact with the platform. As a result, the sequential user feedback data collected by Netflix can be viewed as generated by a sequence of independent but possibly non-identical distributions.\footnote{See \cite{Cao2014} for a discussion on the issue of non-IIDness in online recommender systems.} In this case, Netflix is motivated to use its data in a robust manner to guard against the concern that the data could be more from one type of user, whereas future users that determine payoffs could be more of a different type. 

Regardless of what data Netflix/the DM observes, they can always rely on their initial belief to make the benchmark decision, which is robust against all initially possible DGPs. Notice, because the experiments are independent, the benchmark decision does not depend on the data at all.\footnote{With independence, the set of conditional distributions over future outcomes is solely determined by the set of DGPs regardless of the realized outcomes. Independence is assumed because this fact eases the exposition. In more general environments, the benchmark should be understood as the robust decision under full Bayesian or prior-by-prior updating \citep{pires2002rule}. Namely, the DM applies Bayes' rule to update every DGP in the initial set conditioning on the sample data to obtain a set of posteriors over future outcomes. In this case, all results either directly apply or generalize under standard conditions, see Appendix \ref{apx: dependent}.} As a result, the benchmark decision ensures robustness, but at the cost of ignoring any information that might lead to  better decisions. 

The goal of this paper is to develop a strategy for using data that addresses both the concern for robustness and the desire for improvement. To preview the final output, the proposed strategy ensures that, in a broad class of decision problems, no matter which DGP in the initial set is true, a sufficient amount of data will almost surely lead the data-revised decision to objectively improve upon the benchmark decision, and sometimes strictly so. With finite samples, the same holds with a pre-specified probability. Hence, this strategy enables the DM to achieve higher expected payoffs under the truth while preserving robustness, i.e., failing to do so only with either zero or a small pre-specified probability. 

What must be true for such a strategy? The following example illustrates a key observation of this paper. 

\begin{example}[Introductory Example] 
	Suppose Netflix is deciding how to recommend a movie to a population of users who might either like it (thumbs-Up) or dislike it (thumbs-Down). Let $\{U, D\}$ denote these two possible outcomes observed after a recommendation. Let full state space be $\{U, D\}^{\infty}$, i.e., the infinite Cartesian product of these outcomes. Denote any probability distribution over $\{U, D\}$ by the probability of $U$. A data-generating process, $P$, can be written as $P = P_{1} \times P_{2} \times \cdots$ with $P_{i} \in [0,1]$. Suppose Netflix's initial knowledge about the users is given by the set $\{(1/3)^{\infty}\} \cup \{3/5, 1\}^{\infty}$, where $(1/3)^{\infty}$ denotes an i.i.d. sequence with marginal probability $1/3$ and the other set is defined as
	\begin{equation*}
		\{3/5, 1\}^{\infty} \equiv \{P: P_{i} \in \{3/5, 1\}\}.
	\end{equation*}
	Intuitively, users share some similarities: they either all like the movie with the same low probability $1/3$, or all like it with high but possibly heterogeneous probabilities, $3/5$ or $1$. Suppose Netflix observes feedback from having recommended the movie to $N$ users and must decide how to recommend it to future users. 
	
	\textbf{Decision Problem \Rmnum{1}.} Suppose Netflix can recommend the movie either \textit{aggressively} (a) or \textit{mildly} (m).\footnote{For instance, the movie can be recommended either on top of the front page or further down the list. All payoffs can be interpreted as cardinal measures of how much the recommendation affects users' likelihood of interacting with Netflix.} The aggressive recommendation yields a higher payoff than a mild one if users like the movie, but is worse otherwise. Specifically, let $a(U) = 2$, $a(D) = -1$, $m(U) = 1/2$, and $m(D) = 0$ be the state-contingent payoffs. Within the initial set of DGPs, the worst-case DGP for both types of recommendation is $(1/3)^{\infty}$, since it gives the lowest probability of getting a $U$. Thus, the benchmark decision is to choose $m$. 
	
	With the sample data, Netflix can make a data-revised decision by revising the initial set of DGPs. A commonly used method is maximum likelihood updating, which revises the initial set to the subset of DGPs that maximize the likelihood of observing the data.\footnote{Proposed and axiomatized in \citet{GILBOA199333} and \citet{CHENG2022102587}, this updating rule is applied in \citet{Epstein2007} for studying problems in a similar setting.} When the true DGP is $(1/3)^{\infty}$, with $N$ sufficiently large, Netflix will almost surely observe sample data with an empirical frequency of $U$ close to $1/3$. The likelihood of any such sample, however, is never maximized by the true DGP in the initial set. Indeed, given any such sequence of outcomes, one can always find a DGP in the set $\{3/5, 1\}^{\infty}$ that assigns probability 1 to outcome $U$ whenever $U$ is observed and probability $3/5$ to outcome $U$ whenever $D$ is observed. In this case, notice that for all $N$, 
	\begin{equation*}
		(1/3)^{N/3} \times (2/3)^{2N/3} < 1^{N/3} \times (2/5)^{2N/3},
	\end{equation*}
	i.e., the likelihood of observing the given sample under this heterogeneous DGP is always strictly higher than under the true one. Thus, with maximum likelihood updating, Netflix will asymptotically almost surely rule out the true DGP from their data-revised set.\footnote{While the present paper, to my knowledge, is the first to make such an observation in the context of maximum likelihood updating, it is intrinsically related to the infamous incidental parameter problem in making estimates using maximum likelihood, discovered by \cite{Neyman1948} (see \cite{LANCASTER2000391} for a review), and more broadly, the problem of overfitting. The same observation continues to hold for likelihood-based updating rules, such as those proposed in \citet{Epstein2007} and \citet{CHENG2022102587}.} As the worst-case probability of $U$ in the data-revised set increases to $3/5$, the aggressive recommendation becomes Netflix's data-revised decision. However, this decision gives Netflix an expected payoff of $0$ under the true DGP, which is strictly lower than the benchmark $(1/6)$. Therefore, with maximum likelihood updating, Netflix will almost surely be worse off than under their benchmark decision when the true DGP is $(1/3)^{\infty}$.
\end{example}

The preceding example shows that ruling out the true DGP can lead to objectively worse decisions. Theorem \ref{thm: truth_accommodating_refinement_basic} formalizes this observation by showing that accommodating (i.e., containing, with the necessary technical generalizations) the true DGP when revising the initial set is necessary for guaranteeing objective improvements. This necessity holds even when restricting attention to \emph{basic} decision problems, those consist of two alternatives, one of which yields a constant payoff. Since likelihood-based rules fail this criterion, this paper develops a simple \textbf{empirical distribution method} that accommodates the truth, as illustrated below. In the introductory example, this also proves sufficient for objective improvement.

\begin{example}[Introductory Example Continued]
	Under the empirical distribution method, Netflix revises the initial set to include DGPs whose average mixture of sample marginals, i.e., $1/N\sum_{i=1}^{N} P_{i}$ in this example, is close enough to the empirical distribution of observed outcomes. Kolmogorov's strong law of large numbers ensures that the true DGP will be retained in such data-revised sets asymptotically almost surely. More importantly, whenever the true DGP is retained, the data-revised decision is $m$ if the true DGP is $(1/3)^{\infty}$, and $a$ if the true DGP belongs to the set $\{3/5, 1\}^{\infty}$. In all cases, it is either the same or strictly better than the benchmark decision, thereby achieving objective improvement.
\end{example}

Notice that in Decision Problem \Rmnum{1}, both alternatives rank the two states $U$ and $D$ the same in terms of payoffs. Hence, a \emph{higher} probability of $U$ leads to a \emph{higher} expected payoff under each alternative. I define a decision problem (i.e., a set of alternatives) that possesses this property as a \textit{monotone decision problem}.\footnote{This terminology stems from the definition of a ``monotone decision problem" in \citet{ATHEY2018101}.} More specifically, all alternatives in a monotone decision problem are positive affine transformations of a common payoff function. In this class of problems, accommodating the truth is both necessary and sufficient for achieving objective improvement (Theorem \ref{thm_monotone}). Monotone decision problems arise naturally in economic contexts; Section \ref{characterization} further illustrates how this result may be applied to canonical principal-agent models. 

However, this sufficiency does not extend beyond monotone decision problems (Theorem \ref{thm_monotone_converse}). In other words, such monotonicity is \textit{the} property that characterizes when accommodating the truth suffices for objective improvement. 

In searching for objective improvements in all decision problems, Theorem \ref{thm_impossible} and Corollary \ref{cor_atmost_one} together provide an impossibility result: It is achieved if and only if the true DGP is uniquely identified from the data. Such identification, however, is generally infeasible in non-identical environments.\footnote{As in the example, for any sample size $N$, there always exist multiple DGPs in the set $\{3/5, 1\}^{\infty}$ such that they are the same up to the $N$-th experiment, but have different marginals over future outcomes.} Nevertheless, accommodating the truth still provides a weaker yet meaningful guarantee of objective improvement in all decision problems (Proposition \ref{prop: improving_worst_case}). In particular, it ensures that the decision maker would never prefer to forgo the opportunity to revise decisions using data in exchange for the certainty-equivalent payoff of the benchmark decision.

To summarize, all these results jointly provide three decision-theoretic justifications for accommodating the truth in data revision: (1) it is necessary for objective improvements, (2) it is sufficient in monotone decision problems; and (3) it provides a weaker but still useful guarantee in all decision problems. The paper then proceeds to develop practical methods that accommodate the truth in non-identical environments.

\bigskip 

First, I show that the simple empirical distribution method accommodates the truth asymptotically almost surely (Theorem \ref{thm_asymptotic}). 

Additional concerns arise when considering data samples of bounded size. The standard finite-sample approach is to develop methods that accommodate the truth with a pre-specified probability. This entails making statistical inferences from independent but non-identical samples, which can be computationally challenging. To address this difficulty, I develop an easy-to-implement method, the \emph{augmented i.i.d. test}, which operates as a simple augmentation to the standard practice of constructing confidence regions from i.i.d. data. Theorem \ref{thm_confident} establishes that this method accommodates the truth with at least the required pre-specified probability. For a concrete illustration, Section \ref{Bernoulli} presents a Bernoulli model with ambiguous nuisance parameters and shows that the data-revised sets given by the augmented i.i.d. test are simple extensions of the \textit{Wilson Confidence Intervals} \citep{Wilson1927}. 

Similar non-identical decision environments have been adopted in the literature to study, for instance, social learning with heterogeneous individuals \citep{reshidi2020information, Chen2019}, dynamic portfolio choice under ambiguous idiosyncratic shocks \citep{Epstein2007}, and estimating parameters of incomplete models \citep{Epstein2016}. In all these models, the revision rules developed in this paper can provide new perspectives and often distinct predictions. As an illustration, in Section \ref{Gaussian}, I consider the model in \citet{reshidi2020information}. They assume the DM applies prior-by-prior updating and show there can be non-vanishing ambiguity. Proposition \ref{prop_gaussian} here shows all that ambiguity vanishes asymptotically if the DM instead applies the revision rules proposed here. In fact, learning under the new revision rules is more effective compared to commonly used updating rules and is guaranteed to be correct. 

Finally, I briefly summarize the main contributions of this paper to the literature on robust statistical decisions \citep{wald1950statistical, Watson2016, Hansen2016, manski2021econometrics}. A more detailed discussion is provided in Section \ref{literature}. Most papers in this literature focus on only the ``data-revised decision'' in their decision context. The main innovation here is to explicitly compare the data-revised decision with the benchmark decision that could have been made without using data. This comparison turns out to be not always obvious. Indeed, I show that it is impossible to guarantee the data-revised decision to be always better whenever the true DGP is not uniquely identified. In light of this impossibility, the present paper contributes to the literature by (1) characterizing when the simple criterion of accommodating the truth suffices for objective improvement and (2) providing practical inference methods that achieve it in both asymptotic and finite-sample settings.

\textbf{Outline.} Section \ref{setup} introduces the decision environment. Section \ref{characterization} characterizes when and how objective improvements can be achieved. Section \ref{revision} defines and studies the proposed revision rules. Section \ref{application} presents two applications with parametric models. Section \ref{literature} provides a further discussion on related literature. All proofs are collected in the Appendix. 

\section{The Decision Environment}\label{setup}
There is a countably infinite sequence of random experiments that all have the same finite set of outcomes $S$ with generic element $s$.\footnote{The finiteness assumption is made for simplicity. Section \ref{sec: monotone_moment} and Appendix \ref{apx: dependent} provide discussions on how to extend the results in this paper to infinite state spaces.} The experiments are ordered and indexed by the set $\N = \{1,2,\cdots\}$. The \emph{full state space} is denoted by $\Omega = S^{\infty} \equiv \prod_{i} S_{i}$ with generic element $\omega$. Let $\Sigma$ denote the discrete $\sigma$-algebra on $S$ and $\Sigma^{\infty}$ the product $\sigma$-algebra on $\Omega$. 

For any given sample size $N \in \N$, let $S^{N} \equiv  \prod_{i=1}^{N}S_{i}$ denote the \emph{sample states}. A decision-maker (DM) observes \emph{sample data} $\omega^{N} \in S^{N}$ and makes a decision whose payoff depends only on future unrealized experiments. Let $\Omega_{N} \equiv \prod_{i=N+1}^{\infty} S_{i} \equiv S_{N}^{\infty}$ denote the \emph{future states} with generic element $\tilde{\omega}_{N}$. Define $\Sigma_{N}$ and $\Sigma_{N}^{\infty}$ similarly. 

The DM's decision is the choice from a set of alternatives, named \emph{acts}. An act is defined as a bounded $\Sigma_{N}^{\infty}$-measurable function, $f: \Omega_{N} \rightarrow \R$, that maps future states to payoffs (measured in utilities). Let $\mathcal{F}$ be the space of all acts endowed with the product topology.\footnote{For simplicity, I suppress the dependence of $\mathcal{F}$ on $N$. This is without loss of generality since any act on $\Omega_{N}$ can be isomorphically defined on $\Omega_{M}$ for all $M \in \N$.} An act is \emph{finitely-based} if it depends only on finitely many experiments (i.e., not on tail events). Let $x \in \R$ denote a constant act that gives the same payoff $x$ in all future states. A \emph{decision problem} $D$ is a nonempty and compact subset of $\mathcal{F}$. Let $\mathcal{D}$ denote the collection of all decision problems. Say that a decision problem $D$ is \emph{binary} if $|D| = 2$. 

I call a sequence of independent but possibly non-identical distributions over outcomes a \emph{data-generating process (DGP)}. Formally, let $\Delta(\Omega)$ denote the set of all countably additive probability measures on $(\Omega,\Sigma^\infty)$, endowed with the topology of setwise convergence. The collection of all DGPs is the subset $\Delta_{\text{indep}}(\Omega)= \prod_{i=1}^\infty \Delta(S_i)\subseteq\Delta(\Omega)$. On $\Delta_{\text{indep}}(\Omega)$, the subspace topology inherited from $\Delta(\Omega)$ under setwise convergence coincides with the product topology on $\prod_{i=1}^\infty \Delta(S_i)$. For each $P \in \Delta_{\text{indep}}(\Omega)$, write $P_i$ for its marginal distribution on $S_i$, $P^N$ for its joint marginal on $S^N$, and $P_N^\infty$ for its marginal on $\Omega_N = S_N^{\infty}$.

Suppose the DM knows there is an \emph{initial set} $\mathcal{P}$ of possible DGPs. When fixing some sample data $\omega^{N}$, let $P^{*}$ denote the \emph{true DGP} that generates the data and governs the future states. The DM's initial knowledge is ``correctly specified'', meaning that it always contains the true DGP. Using the terminology from Bayesian learning literature \citep{Kalai1993}, the DM's initial knowledge contains a ``grain of truth'': 

\begin{assumption}\label{assumption1}
	$P^{*} \in \mathcal{P}$. 
\end{assumption}

What this assumption also means is that every DGP in the initial set could be the true one governing the experiments. In addition, assume the set $\mathcal{P}$ is compact and all $P \in \mathcal{P}$ have full support.\footnote{This is a standard assumption for studying the update of a set of probability measures. It is also used when establishing a Central Limit Theorem in this environment. See the proof of Lemma \ref{lem1} for details.} Let $\mathcal{P}_{i}$, $\mathcal{P}^{N}$ and $\mathcal{P}_{N}^{\infty}$ denote the sets of their corresponding marginals. 

Given a decision problem $D$, the DM can make a robust decision by using the initial set of DGPs. Formally, let the \emph{benchmark decision}, denoted by $c(D) \in D$, be given by
\begin{equation*}
	c(D) \equiv \arg\max\limits_{f \in D}\min\limits_{P \in \mathcal{P}} \int_{\Omega_{N}} f(\tilde{\omega}_{N})dP_{N}^{\infty}(\tilde{\omega}_{N}) = \arg\max\limits_{f \in D}\min\limits_{P_{N}^{\infty} \in \overline{co}(\mathcal{P}_{N}^{\infty})} \int_{\Omega_{N}}   f(\tilde{\omega}_{N})dP_{N}^{\infty}(\tilde{\omega}_{N}),
\end{equation*}
where $\overline{co}(\mathcal{P}_{N}^{\infty})$ denotes the closed and convex hull of $\mathcal{P}_{N}^{\infty}$. The ``min'' is well-defined as $\mathcal{P}$ is compact. The equality follows since the minimum can be achieved at an extreme point and thus belongs to $\mathcal{P}$. Moreover, $c(D)$ is defined as a singleton by imposing an arbitrary (but fixed) tie-breaking rule. 

With sample data $\omega^{N}$, the DM can also choose to revise their initial set to a \emph{data-revised set} of DGPs. Let $\mathcal{P}(\omega^{N}) \subseteq \Delta_{indep}(\Omega)$ denote the revised set. Given Assumption \ref{assumption1}, further let $\mathcal{P}(\omega^{N}) \subseteq \mathcal{P}$ as there is no need to consider DGPs outside $\mathcal{P}$. Importantly, the data-revised set is assumed to depend only on the initial set $\mathcal{P}$ and sample data $\omega^{N}$ but not on $D$, i.e., the specific decision problem considered. In other words, the DM's revision rule is purely ``data-based''. In the literature on belief updating, this is known as consequentialism, see \citet{hanany2007updating, hanany2009updating} and \citet{Siniscalchi2009-SINTOO-5} for discussions. 

Let $c(D, \omega^{N}) \in D$ denote the DM's \emph{data-revised decision} given (the closed and convex hull of) their data-revised set $\mathcal{P}(\omega^{N})$, i.e., 
\begin{equation*}
	c(D, \omega^{N}) \equiv \arg\max\limits_{f \in D}\min\limits_{P_{N}^{\infty} \in \overline{co}(\mathcal{P}(\omega^{N})_{N}^{\infty})} \int_{\Omega_{N}} f(\tilde{\omega}_{N})dP_{N}^{\infty}(\tilde{\omega}_{N}),
\end{equation*}
where the ``min'' is well defined as $\overline{\mathrm{co}}\big(\mathcal{P}(\omega^{N})_{N}^{\infty}\big)$ is compact and $c(D,\omega^{N})$ is also defined to be a singleton by imposing a tie-breaking rule, subject to the following consistency requirement:
\begin{assumption}\label{assumption4}
If $c(D)$ is among the maximizers in the data-revised program, then $c(D,\omega^{N})=c(D)$. 
Otherwise, the tie-breaking can be arbitrary.
\end{assumption}

Assumption \ref{assumption4} makes sure that if the DM's data-revised set coincides with their initial set, then their data-revised decision is the same as the benchmark. This rules out uninteresting complications when the DM's decisions are different only because of different tie-breakings.\footnote{Also notice both decisions are defined as deterministic choices from $D$, which seems to rule out randomizations from the DM's decisions. This is, in fact, more general as one can explicitly add those randomizations as acts and it will just be another well-defined decision problem. The current definition allows for richer decision patterns: By imposing different tie-breaking rules, it allows the DM to have different preferences in terms of whether randomization hedges against ambiguity. See \cite{Saito2015} and \cite{Ke2020} for discussions and characterizations of such preferences.} 

Importantly, the DM's \emph{objective payoff} from an act is determined by the true DGP $P^{*}$ that actually governs the future. To make this explicit, let 
\begin{equation*}
	W(f, P^{*}) = \int_{\Omega_{N}} f(\tilde{\omega}_{N}) dP_{N}^{*\infty}(\tilde{\omega}_{N})
\end{equation*}
denote the expected payoff of act $f$ evaluated under the true DGP $P^{*}$. For any decision problem $D$, the DM's objective payoffs from their benchmark and data-revised decisions are then $W(c(D), P^{*}) $ and $W(c(D,\omega^{N}), P^{*})$, respectively. 

Because the DM's decisions depend only on future states and are made using the robust criterion, it is useful to define the following notions that suitably generalize set inclusions: 
\begin{definition}\label{def: accommodating and refinement}
	The data-revised set $\mathcal{P}(\omega^{N})$ \textbf{accommodates} a DGP $P$ if 
	\begin{equation*}
		P_{N}^{\infty} \in \overline{co}(\mathcal{P}(\omega^{N})_{N}^{\infty}). 
	\end{equation*}
	The data-revised set $\mathcal{P}(\omega^{N})$ \textbf{refines} the initial set $\mathcal{P}$ if 
	\begin{equation*}
		\overline{co}(\mathcal{P}(\omega^{N})_{N}^{\infty}) \subsetneqq \overline{co}(\mathcal{P}_{N}^{\infty}).
	\end{equation*}
	Say that $\mathcal{P}(\omega^{N})$ is a \textbf{truth-accommodating refinement} if it accommodates the true DGP and refines the initial set. 
\end{definition}

Notice that if $P^{*} \in \mathcal{P}(\omega^{N})$, then $\mathcal{P}(\omega^{N})$ accommodates $P^{*}$. Moreover, $\mathcal{P}(\omega^{N})$ refines $\mathcal{P}$ only if $\mathcal{P}(\omega^{N})$ is a proper subset of $\mathcal{P}$. In practice, since a convex combination of DGPs in $\Delta_{indep}(\Omega)$ remains in $\Delta_{indep}(\Omega)$ only when the DGPs share identical marginals in all but one experiment, a truth-accommodating refinement is thus essentially $(P^{*})_{N}^{\infty} \in \mathcal{P}(\omega^{N})_{N}^{\infty} \subsetneqq \mathcal{P}_{N}^{\infty}$.

\section{Objective Improvements}\label{characterization}
Fix an initial set $\mathcal{P}$, some sample data $\omega^{N}$ generated by the true DGP $P^{*} \in \mathcal{P}$. This section investigates what guarantees objective improvements across different decision problems. To this end, consider the following definition. 

\begin{definition}\label{def: objective_improvement_class}
    Let $\mathcal{C} \subseteq \mathcal{D}$ denote a \textbf{class} of decision problems. A data revision is said to provide \textbf{objective improvement in $\mathcal{C}$} if, for all $D \in \mathcal{C}$, 
	\begin{equation*}
		W(c(D, \omega^{N}), P^{*}) \geq W(c(D), P^{*}),
	\end{equation*}
	and the inequality is strict for some $D \in \mathcal{C}$. 
\end{definition}

Among all decision problems, the simplest possible form is the choice between an uncertain and a constant act. Such a canonical form is often used to model, for example, the decision of whether or not to approve a new drug, implement a new policy, convict a defendant, and invest in an asset. I call such decision problems, i.e., binary with a constant act, \emph{basic decision problems}. Arguably, basic decision problems are the building blocks of more complex decision problems, thus any general enough class of decision problems should include them as special cases. Guaranteeing objective improvements at least in all basic decision problems is a reasonable minimal requirement for any data revision. The following result identifies a necessary condition for this guarantee: The data-revised set must accommodate the true DGP.

\begin{theorem}\label{thm: truth_accommodating_refinement_basic}
	If the data-revised set $\mathcal{P}(\omega^{N})$ does not accommodate the true DGP $P^{*}$, then there exists a basic decision problem for which the data-revised decision is objectively worse than the benchmark decision, i.e., $W(c(D, \omega^{N}), P^{*}) < W(c(D), P^{*})$.
\end{theorem}

The proof of Theorem \ref{thm: truth_accommodating_refinement_basic} relies on a standard separating hyperplane argument by noticing that $P^{*} \in \mathcal{P}$ (Assumption \ref{assumption1}), but $(P^{*})_{N}^{\infty} \notin \overline{co}(\mathcal{P}(\omega^{N})_{N}^{\infty})$ (Definition \ref{def: accommodating and refinement}). The key takeaway is that if the data-revised set fails to accommodate the true DGP, then one can easily construct a basic decision problem for which the data-revised decision is objectively worse. In this sense, Theorem \ref{thm: truth_accommodating_refinement_basic} highlights that accommodating the true DGP is a necessary condition for achieving objective improvement in any class of decision problems including basic decision problems as special cases.

A truth-accommodating refinement, a slight strengthening of this necessary condition, can be shown to be also sufficient for objective improvement in basic decision problems. The next section establishes this result by characterizing the essentially largest class of decision problems for which this sufficiency holds, a class that includes all basic decision problems.

\subsection{Monotone Decision Problems}\label{sec: monotone_decision_problems}
What structural feature of a decision problem makes a truth-accommodating refinement sufficient for objective improvement? To build intuition, consider a \emph{betting decision problem}, where every available act is a bet on the \emph{same} event. Formally, in a betting decision problem, there exists $A \subseteq \Omega_{N}$ such that for each $f$, for all $\tilde{\omega} \in A$ and $\tilde{\omega}' \in A^{c}$, 
\begin{equation*}
	f(\tilde{\omega}) = f(A) \geq f(A^{c}) = f(\tilde{\omega}').
\end{equation*}
In words, all acts rank outcomes in the same way across the two events, $A$ and $A^{c}$, but differ in their payoff tradeoffs. As a result, a higher belief in $A$ (i.e., a greater probability assigned to $A$) leads to a higher expected payoff of every act. Moreover, all acts can be ordered so that higher acts are optimal under higher beliefs. This is in line with the definition of monotone decision problems in \citet{ATHEY2018101},\footnote{Their definition says that the optimal act is monotone in signal $x$, which corresponds to a posterior belief over states. The additional property here is that it also leads to a higher expected payoff.}  which highlights that both payoffs and choices move in the same direction as beliefs.

This form of monotonicity is precisely what makes a truth-accommodating refinement sufficient for objective improvement. When choosing how much to bet on the event $A$, the worst-case DGP is always the one assigning the lowest probability to $A$. A truth-accommodating refinement raises this worst-case probability while keeping it below the true probability. Monotonicity then guarantees that the data-revised decision moves closer to the exact optimal act under $P^{*}$, thereby delivering a higher objective payoff than the benchmark. 

Motivated by this intuition, I now generalize the betting problem to higher dimensions by defining a class of \emph{monotone decision problems}.

\begin{definition}\label{def_monotone}
    A decision problem $D$ is called a \textbf{monotone decision problem} if there exists an act $g \in \mathcal{F}$ (not necessarily contained in $D$) such that, for every $f \in D$, 
    \begin{equation*}
        f = \lambda_{f} g + c_{f}, 
    \end{equation*}
    for some $\lambda_{f} \geq 0$ and $c_{f} \in \R$. Let $\mathcal{D}_{m}$ denote the class of monotone decision problems. 
\end{definition}

Intuitively, all acts in a monotone decision problem are non-negative affine transformations of a common reference act $g$. Hence they induce the same ranking over future states and satisfy
\begin{equation*}
    W(f, P) = \lambda_{f} W(g, P) + c_{f}, \quad \forall f \in D.
\end{equation*}
Thus the DM's evaluation reduces to the single statistic $W(g,P)$: a higher belief in $g$ (i.e., a greater expected payoff of $g$) raises the expected payoff of every act. This generalizes the monotonicity observed in betting problems. Importantly, the ``monotone'' label encodes a \emph{one-dimensional} structure: although the underlying state space may be high-dimensional, all payoff-relevant variation collapses to the single expectation of $g$. From a geometric perspective, after subtracting the constant, all acts are aligned in the same direction. Thus, choosing among these acts is essentially trading off between their sensitivity to beliefs $(\lambda_{f})$ against their guaranteed payoffs $(c_{f})$, just as in betting decision problems one trades off payoffs across two events. 

Basic decision problems are trivially monotone, since the only non-constant act can be taken as the reference act. It thus follows from Theorem \ref{thm: truth_accommodating_refinement_basic} that accommodating the true DGP is necessary for objective improvement in monotone decision problems. The next result establishes sufficiency.

\begin{theorem}\label{thm_monotone}
	A data revision provides objective improvement in monotone decision problems if and only if the data-revised set is a truth-accommodating refinement. 
\end{theorem}

Theorem \ref{thm_monotone} delivers a clear and powerful message: a truth-accommodating refinement is sufficient to guarantee objective improvement in monotone decision problems, regardless of the true DGP. Its proof builds on the intuition developed in betting decision problems, but the extension is substantial: monotone decision problems constitute a broader class that encompasses a wide range of economically relevant environments. To illustrate this concretely, I next show that the problem of choosing among linear contracts in a canonical principal-agent model is a monotone decision problem. Consequently, Theorem \ref{thm_monotone} applies directly and yields new insights.

\begin{example}\label{ex:linear-contract} [Improving Linear Contracts with Data]  
A principal hires an agent to work on a project whose output is given by $y = \theta e + \epsilon$, where $\theta \in \mathbb{R}_{+}$ denotes the agent's productivity, $e \in \mathbb{R}_{+}$ the agent's effort, and $\epsilon$ a random noise with $\mathbb{E}[\epsilon]=0$. The agent's effort cost is $c(e)=k e^{2}/2$ for some $k>0$. 

The principal faces uncertainty about the agent's productivity. Let $\Theta \subseteq \mathbb{R}_{+}$ denote the set of possible productivity levels and $\mathcal{P}\subseteq \Delta(\Theta)$ the principal's initial belief set. The principal's objective is to choose a robustly optimal \textbf{linear contract} to offer the agent, given $\mathcal{P}$ and a possible revision using data from past projects.

A linear contract specifies a base wage $\alpha \in \mathbb{R}$ and a share $\beta \in [0,1]$ of output. Given $(\alpha,\beta)$, the agent chooses effort $e$ to maximize $\mathbb{E}_{\epsilon} \big[\alpha + \beta(\theta e + \epsilon) - k e^{2}/2\big]$, yielding the optimal effort $e^{*}(\theta,\alpha,\beta)=\beta\theta/k \geq 0$. Thus, given $(\alpha,\beta)$ and $\theta$, the agent's expected payoff is $U_{A}(\theta;\alpha,\beta) = \alpha+\beta^{2}\theta^{2}/2k$ and the principal's expected payoff is 
\begin{align*}
\pi_{(\alpha,\beta)}(\theta) = \mathbb{E}_{\epsilon} \big[(1-\beta)(\theta e^{*}(\theta,\alpha,\beta)+\epsilon)-\alpha\big]
=\frac{\beta(1-\beta)}{k}\theta^{2}-\alpha.
\end{align*}

Let $u_{0}:\Theta\to\mathbb{R}$ denote the agent's outside-option payoff. Under robustness considerations, the principal restricts attention to contracts that satisfy individual rationality (IR) uniformly across all types.\footnote{The uniform IR condition provides one natural way to ensure that the linear-contract problem is monotone. Nevertheless, the monotone structure can also be preserved under alternative IR formulations, provided that the set of participating types remains invariant across contracts.} This pins down contracts to those satisfying
\begin{equation*}
\alpha\ge \sup_{\theta\in\Theta} \left( u_{0}(\theta)-\frac{\beta^{2}\theta^{2}}{2k} \right)\equiv \alpha_{\min}(\beta).
\end{equation*}
Because the principal's payoff is decreasing in $\alpha$, it is without loss to focus on $\alpha=\alpha_{\min}(\beta)$. Thus, the principal's expected payoff from a linear contract with share $\beta$ becomes
\begin{equation*}
\pi_{\beta}(\theta) =\frac{\beta(1-\beta)}{k}\theta^{2}-\alpha_{\min}(\beta) \equiv \lambda_{\beta}\,g(\theta)+c_{\beta}.
\end{equation*}
This implies that choosing among feasible contracts (i.e., among $\beta\in[0,1]$) is to choose among acts $\pi_{\beta}$ that are non-negative affine transformations of the common act $g(\theta) = \theta^{2}$, i.e., a monotone decision problem.

By Theorem \ref{thm_monotone}, if the principal revises the initial belief set $\mathcal{P}$ about the agent's productivity to a truth-accommodating refinement using data, the resulting data-revised contract is guaranteed to yield a higher expected payoff than the benchmark contract regardless of the true productivity distribution.
\end{example}

The linear-contract example illustrates how monotone decision problems can arise in economic settings. The key is to identify a decision problem where all available options can be represented as non-negative affine transformations of a common act. In this example, the assumption of linear contracts is essential for representing the principal's decision problem as monotone. This assumption may be partially justified by the prominent role of linear contracts under robustness concerns \citep{Carroll2015, carroll2016}. Nevertheless, to extend the observations here to more contracting environments, the generalizations identified in the next section are useful.

\subsection{Monotone-Like Decision Problems}\label{sec: monotone_moment}
Theorem \ref{thm_monotone} identifies monotonicity as a sufficient condition for a truth-accommodating refinement to guarantee objective improvement when $\mathcal{P}$ is arbitrary and acts are functions from states to payoffs. Its core intuition, in fact, extends more broadly as additional structure is imposed on either the initial set or the acts. Notice the key force behind Theorem \ref{thm_monotone} is two-fold: (i) all acts share a common worst-case DGP, and (ii) a truth-accommodating refinement shifts the robust decision towards higher payoffs under every DGP in the revised set, thereby ensuring an objective improvement regardless of which DGP is the truth. In this section, I illustrate that these two forces can be found in several alternative formulations that are relevant in applications.

\subsubsection{Monotone Decision Problems Under FOSD}
Fix an order on $\Omega$. Say that $\mathcal{P}$ is \emph{FOSD-comparable} if all DGPs in $\mathcal{P}$ are totally ordered by first-order stochastic dominance with respect to this order. A decision problem $D$ is \textbf{monotone under FOSD} if, under the same order on $\Omega$, every act in $D$ is increasing and for any $f, g \in D$, the difference $f - g$ is monotone (either increasing or decreasing). This requirement is weaker than Definition \ref{def_monotone} as it does not restrict acts to be affine transformation of one another. Nevertheless, it preserves the two forces identified above and thus guarantees objective improvement under truth-accommodating refinements.

\begin{corollary}\label{cor:fosd}
	Fix an order on $\Omega$ and suppose $\mathcal{P}$ is FOSD-comparable. If the data-revised set is a truth-accommodating refinement, then the data revision provides objective improvement in all decision problems that are monotone under FOSD.
\end{corollary}

To see why, notice when $\mathcal{P}$ is FOSD-comparable, all increasing acts share the same worst-case DGP. If $\mathcal{P}(\omega^{N})$ is a truth-accommodating refinement, then the true DGP $P^{*}$ dominates the common worst-case DGP $P_{1}$ in $\mathcal{P}(\omega^{N})$, which in turn dominates the common worst-case DGP $P_{2}$ in $\mathcal{P}$. Let $f = c(D)$ and $g = c(D, \omega^{N})$. It follows that $W(f, P_{2}) \leq W(g, P_{2})$ and $W(g, P_{1}) > W(f, P_{1})$. Let $h = g - f$ and by monotonicity, $h$ must be increasing. Hence $W(h, P^{*}) \geq W(h, P_{1}) > 0$, which implies the desired objective improvement.

\begin{example}[Improving (Non-Linear) Contracts with Data]
Continuing the previous example, but additionally suppose the principal's initial knowledge $\mathcal{P}$ is FOSD-comparable with respect to $\Theta \subseteq \R_{+}$. This additional structure allows the principal to consider more general contracts beyond linear ones while still being able to guarantee objective improvement under truth-accommodating refinements.

Concretely, let $\Theta = [\underline{\theta}, \overline{\theta}] \subseteq \R_{+}$ and suppose the principal expands the contract space to include the following quadratic form of contracts:
\begin{align*}
    w_{(\beta, \gamma)}(y) = \alpha_{\min}(\beta,\gamma) + \beta y + \frac{\gamma}{2} y^{2}, 
\end{align*}
with $\beta \in [0,1]$, $\gamma < (1-\beta)k/\overline{\theta}^{2}$, and $\alpha_{\min}(\beta,\gamma)$ the minimal base wage satisfying uniform IR. For simplicity, let $\epsilon \equiv 0$. The upper bound on $\gamma$ ensures that $e^{*}(\theta, \beta, \gamma) = \beta\theta/(k - \gamma \theta^{2})$ is always well-defined and the principal's payoff, 
\begin{align*}
	\pi_{(\beta, \gamma)}(\theta) = \frac{(1-\beta)\beta\theta^{2}}{k - \gamma \theta^{2}} - \frac{\gamma \beta^{2} \theta^{4}}{2(k - \gamma \theta^{2})^{2}} - \alpha_{\min}(\beta,\gamma),
\end{align*}
is increasing in $\theta$. However, it is not necessarily true that for any two feasible contracts $(\beta, \gamma)$ and $(\beta', \gamma')$, the difference $\pi_{(\beta, \gamma)}(\theta) - \pi_{(\beta', \gamma')}(\theta)$ is monotone in $\theta$. Nevertheless, if the principal's optimal contracts before and after data revision are such that this difference is monotone in $\theta$, then objective improvement can be established by examining only these two contracts: by Corollary \ref{cor:fosd}, objective improvement holds in the decision problem restricted to this pair. Since the principal in fact selects these same two contracts in the full problem, enlarging the contract space to include additional contracts that are not chosen does not affect the objective-improvement conclusion.
\end{example}

\subsubsection{Monotone Moment Decision Problems}
In fields such as information design  \citep{gentzkow2016, Kolotilin2017,dworczak2019a}, among others, decision problems are sometimes modeled as choosing among options whose payoffs depend on the distribution over states only through a scalar \emph{moment} (e.g., the mean). Formally, fix a bounded measurable map $m: \Omega_{N} \rightarrow \R$ and write $m(P) \equiv \mathbb{E}_{P}[m(\omega)]$. A \emph{moment act} is a function $\overline{f}:\R \to \R$ that assigns payoff $\overline{f}(m(P))$ under distribution $P$. Note that while every (state–contingent) act induces an expectation under each $P$, not every moment act corresponds to such a state–contingent act, especially when $\overline f$ is nonlinear in the moment (e.g., $\overline f(x)=x^{2}$).

A decision problem $\overline D$ is a \textbf{monotone moment decision problem} if all acts are increasing moment acts and \emph{single crossing} holds: for all $\overline f,\overline g\in\overline D$, if $\overline f(x)\ge \overline g(x)$ at some $x$, then $\overline f(x')\ge \overline g(x')$ for all $x'\ge x$. As before, these conditions ensure the same two forces identified above. The following corollary summarizes the conclusion.

\begin{corollary}\label{cor:moment}
	If the data-revised set is a truth-accommodating refinement, then the data revision provides objective improvement in all monotone moment decision problems.
\end{corollary}

\begin{example}[Improving Linear Contracts with Data Under Generalized Preferences]
Continuing the linear-contract example, another key assumption that makes the principal's problem monotone is the principal's payoff form. Allowing the principal's preference to depend non-linearly on the expected output but still linear in the expected payment would generally break the monotonicity as in Definition \ref{def_monotone}. Specifically, let $u: \R \to \R$ be an increasing function such that the principal's payoff under a linear contract $(\alpha, \beta) \in \R \times [0,1]$ and a distribution $P \in \Delta(\Theta)$ is 
\begin{equation*}
	\overline{\pi}_{(\alpha,\beta)}(P) = u\left(\frac{\beta}{k} \mathbb{E}_{P}[\theta^{2}]\right) - \frac{\beta^{2}}{k} \mathbb{E}_{P}[\theta^{2}] - \alpha.
\end{equation*}
Let $m(P) = \mathbb{E}_{P} [\theta^{2}]$, then $\overline{\pi}_{(\alpha,\beta)}$ is in the form of a moment act. When $u(\cdot)$ is twice differentiable, $u'(z) \geq 1$ and $u'(z) + z u''(z) \geq 2$ on the relevant range of moments is sufficient for all $\overline{\pi}_{(\alpha,\beta)}$ to be increasing and pairwise single-crossing in $m(P)$. Then by Corollary \ref{cor:moment}, again, a truth-accommodating refinement guarantees objective improvement in choosing among linear contracts.
\end{example}

Monotone decision problems under FOSD and monotone moment decision problems illustrate two distinct yet complementary directions for generalizing the notion of monotonicity while preserving the two key forces underlying Theorem~\ref{thm_monotone}. As emphasized at the beginning, within any specific decision context, whenever these two forces are present, a truth-accommodating refinement guarantees objective improvement.

\subsection{Beyond Monotonicity}
Outside the monotone and monotone-like classes identified above, where a common worst case and a directional monotonicity of payoff differences obtain, the guarantee that truth-accommodating refinements yield objective improvement can break down. The next result shows that this failure is robust: it could arise for arbitrary non-singleton refinements, or arbitrary non-monotone binary decision problems.

\begin{theorem}\label{thm_monotone_converse}
    The following statements are true. 
    \begin{enumerate}[(i)]
        \item For any non-singleton data-revised set $\mathcal{P}(\omega^{N})$ that refines an initial set $\mathcal{P}$, there exists some $P^{*} \in \mathcal{P}(\omega^{N})$ and a non-monotone decision problem $D$ such that
        \begin{equation*}
			W(c(D, \omega^{N}), P^{*}) < W(c(D), P^{*}).
		\end{equation*}
        
        \item For all non-monotone binary decision problem $D$ with finitely-based acts $f_{1}$ and $f_{2}$, if there exists $P \neq P' \in \Delta_{indep}(\Omega)$ such that $W(f_{1}, P) > W(f_{2}, P)$, and $W(f_{1}, P') < W(f_{2}, P')$, then there exists an initial set $\mathcal{P}$, a data-revised set $\mathcal{P}(\omega^{N})$, and a true DGP $P^{*}$ such that $\mathcal{P}(\omega^{N})$ is a truth-accommodating refinement, yet
		\begin{equation*}
			W(c(D, \omega^{N}), P^{*}) < W(c(D), P^{*}).
		\end{equation*}
    \end{enumerate}
\end{theorem}

Theorem \ref{thm_monotone_converse} identifies two different impossibility directions for extending sufficiency beyond monotone decision problems. Part (i) rules out any guarantee based solely on the refinement itself: for every non–singleton refinement $\mathcal{P}(\omega^{N})$ of $\mathcal{P}$ there is a true DGP and a non-monotone problem for which the data-revised choice underperforms the benchmark. Part (ii) emphasizes that there is no particular way of deviating from monotonicity that maintains sufficiency.\footnote{The statement restricts to finitely-based acts to avoid a small caveat involving tail events. See Remark \ref{remark: non_finitely_based} in the proof for details. In addition, similar negative results can be stated for non-binary decision problems as the argument involves only two acts, those that are chosen by the benchmark and data-revised decisions. The presence of other acts only complicates the construction.} In other words, monotonicity is not merely sufficient; it is essentially the boundary for when a truth-accommodating refinement guarantees objective improvement.

To gain some intuition of why without monotonicity, a truth-accommodating refinement could lead to a strictly worse decision, consider the following example. 

\begin{example}[Introductory Example continued] \textbf{Decision Problem \Rmnum{2}.} Netflix decides whether to \textit{include} (i) or \textit{remove} (r) the movie from their recommendations. The key difference from Decision Problem \Rmnum{1} is that the two actions now rank the states in opposite ways: removing the movie yields a higher payoff when users dislike it, whereas including it yields a higher payoff when users like it. Numerically, let the payoffs be $i(U) = 1$, $i(D) = -1$, $r(U) = 0$, and $r(D) = 1$. Because the two alternatives rank the two states differently, this decision problem is non-monotone.

In this case, Netflix's benchmark decision is to choose $r$. If Netflix revises the initial belief using the empirical distribution method described in the introduction, then their data-revised decision will be $r$ when the true DGP is $(1/3)^{\infty}$ and $i$ when the true DGP belongs to the set $\{3/5, 1\}^{\infty}$. However, if the true DGP is $(3/5)^{\infty} \in \{3/5, 1\}^{\infty}$, the expected payoff of $i$ is $1/5$, strictly lower than the expected payoff of $r$, equal to $2/5$. 
\end{example}

In Decision Problem \Rmnum{2}, the two acts rank the two states differently. While a higher probability of $U$ implies a ``higher'' optimal act, it does not necessarily lead to a higher expected payoff. This non-monotonicity invalidates the previous intuition. In this case, the worst-case DGPs for the two acts are those assigning the lowest probability to $U$ and $D$, respectively. While a truth-accommodating refinement still ensures that the lowest probabilities of $U$ and $D$ in the data-revised set are greater than those in the initial set but less than the true probabilities, the different levels of probability increase may lead the data-revised decision to be further away from the exact optimal act.

Decision Problem \Rmnum{2} illustrates one direction where a decision problem can deviate from monotonicity. For all the other directions, see the example in the proof of Theorem \ref{thm_monotone_converse} for an illustration. 

\subsection{Impossibility of Objective Improvement in All Decisions}
Are there stronger conditions than truth-accommodating refinements that could guarantee objective improvement beyond monotone decision problems? This section provides a negative answer when considering the set of all decision problems, i.e., when $\mathcal{C} = \mathcal{D}$. 

Obviously, a sufficient condition is when $P^{*}$ is uniquely identified from $\omega^{N}$ and $\mathcal{P}$ is not a singleton. In this case, by letting  $\mathcal{P}(\omega^{N}) = \{P^{*}\}$, the DM's data-revised decision is exactly optimal against $P^{*}$, thus always improves. However, revising a non-singleton initial set to a singleton set containing the true DGP is not always feasible, especially when the possible DGPs can be non-identical. When the data-revised set is not a singleton, the following theorem shows that objective improvement in all decision problems requires the true DGP to be effectively uniquely identified. 

\begin{theorem}\label{thm_impossible}
	A data revision provides objective improvement in all decision problems if and only if there exists $\alpha \in (0,1]$ such that 
	\begin{equation*}
		\overline{co}(\mathcal{P}(\omega^{N})_{N}^{\infty}) = \alpha P^{*\infty}_{N} + (1-\alpha)\overline{co}(\mathcal{P}_{N}^{\infty}). 
	\end{equation*}
\end{theorem}

The condition $\overline{co}(\mathcal{P}(\omega^{N})_{N}^{\infty}) = \alpha P^{*\infty}_{N} + (1-\alpha)\overline{co}(\mathcal{P}_{N}^{\infty})$ says that, after taking the closed convex hull of future marginals, the data-revised set needs to be a convex combination of the initial set and the true DGP. Observe that the only way to form such a data-revised set requires knowing exactly what $P^{*}$ is. But if it is the case, the DM should just let $\{P^{*}\}$ be their data-revised set. On the other hand, if $\overline{co}(\mathcal{P}(\omega^{N})_{N}^{\infty}) = \alpha P^{*\infty}_{N} + (1-\alpha)\overline{co}(\mathcal{P}_{N}^{\infty})$ for some $\alpha \in (0,1]$, then the same cannot hold for any other $P$ and $\alpha$ whenever $P_{N}^{\infty} \neq P_{N}^{*\infty}$. This can be seen by considering the probability of any arbitrary event. This impossibility thus leads to the following corollary. 

\begin{corollary}\label{cor_atmost_one}
	Any given data-revised set can provide objective improvement in all decision problems under at most one DGP (up to the same marginal over future states). 
\end{corollary}

Given Corollary \ref{cor_atmost_one}, if there are multiple DGPs the DM can by no means distinguish using the sample data (like the ones in the introductory example), then no data-revised set would be able to guarantee objective improvement in all decision problems under \textit{all of them}.

Crucially, a truth-accommodating refinement does not have this issue: A data-revised set continues to be a truth-accommodating refinement no matter which DGP accommodated by it turns out to be the truth. Therefore, Theorem \ref{thm_monotone} indeed further implies that the objective improvement can be guaranteed simultaneously under multiple DGPs, contrasting to the conclusion in Corollary \ref{cor_atmost_one}. 

In addition, a truth-accommodating refinement can also provide a weaker improvement guarantee in all decision problems: 

\begin{proposition}\label{prop: improving_worst_case}
    If the data-revised set $\mathcal{P}(\omega^{N})$ is a truth-accommodating refinement, then for all $D \in \mathcal{D}$,
    \begin{equation*}
		W(c(D, \omega^{N}), P^{*}) \geq \min\limits_{P \in \mathcal{P}} W(c(D), P), 
	\end{equation*}
	and the inequality is strict for some $D \in \mathcal{D}$. 
\end{proposition}

In words, the expected payoff from the data-revised decision is never lower than the guaranteed payoff of the benchmark decision. Hence, whenever the revision rule constitutes a truth-accommodating refinement, the DM would never prefer to forgo the opportunity to revise their decision using data in exchange for the benchmark's certainty-equivalent payoff. This improvement guarantee is relatively weak but not trivial: a misled data-revised decision could be objectively worse than this certainty equivalent. Truth accommodation ensures that it cannot happen. In other words, a truth-accommodating refinement guarantees that learning from data is always valuable relative to receiving the ex-ante certainty equivalent payoff, providing an additional rationale for adopting revision rules that accommodate the truth.

As a final remark, all results in this section also apply when comparing two data-revised sets, say $\mathcal{P}_{1}(\omega^{N})$ and $\mathcal{P}_{2}(\omega^{N})$. When both sets accommodate the truth and $\mathcal{P}_{1}(\omega^{N})$ is a refinement of $\mathcal{P}_{2}(\omega^{N})$, all conclusions about objective improvement holds by viewing $\mathcal{P}_{1}(\omega^{N})$ as a truth-accommodating refinement of $\mathcal{P}_{2}(\omega^{N})$.

\section{Revision Rules that Accommodate the Truth}\label{revision}
As established in the previous section, for improving robust decisions, it is necessary and sometimes sufficient to accommodate the true DGP in revising the initial set (the refinement part only guarantees it to be sometimes strict). Define a \emph{revision rule} as the mapping from an initial set $\mathcal{P}$ and sample data $\omega^{N}$ to a data-revised set $\mathcal{P}(\omega^{N})$. When the sample size is unbounded, the following definition formalizes a notion of accommodating the truth in the asymptotic sense. 

\begin{definition}\label{def_asymptotic}
	A revision rule \textbf{accommodates the truth asymptotically almost surely} if, for any initial set $\mathcal{P}$, for every $P^{*} \in \mathcal{P}$, and for $P^{*}$-almost every $\omega \in \Omega$, there exists $\bar{N}(\omega)$ such that, for all $N \geq \bar{N}(\omega)$, the data-revised set $\mathcal{P}(\omega^{N})$ accommodates the DGP $P^{*}$. 
\end{definition}

This definition says that, regardless of which possible DGP governs the uncertainty, the revision rule ensures that, with a sufficient amount of data, the data-revised set accommodates the true DGP almost surely. Therefore, whenever accommodating the truth is sufficient for objective improvement, such improvements are also achieved asymptotically almost surely.

Likelihood-based rules have been shown in the introductory example to violate this property. Here, I propose a revision rule based on empirical distributions: For any sample data $\omega^{N}$, let $\boldsymbol\Phi(\omega^{N}) \in \Delta(S)$ denote the \emph{empirical distribution}, i.e., for any outcome $s \in S$, $\boldsymbol\Phi(\omega^{N})(s)  \equiv  N^{-1}\sum_{i=1}^{N} I\{\omega_{i} = s\}.$ For any $P \in \Delta_{indep}(\Omega)$ and for any $N$, the \emph{average of sample marginals}, $\bar{P}^{N} \in \Delta(S)$, is defined to be the distribution over $S$ given by the average mixture of the marginal distributions over each component of the sample states, i.e., $\bar{P}^{N} \equiv N^{-1}\sum_{i=1}^{N}P_{i}$. For any $p,q\in \Delta(S)$, let $\rho(p,q)$ denote the sup-norm distance.

\begin{definition}\label{def_empirical_distribution}
	The data-revised sets are obtained by the \textbf{empirical distribution method} if, for some pre-specified $\epsilon > 0$ and for every $\omega^{N}$, 
	\begin{equation*}
		\mathcal{P}(\omega^{N})  = \left\{P\in \mathcal{P}: \rho (\bar{P}^{N}, \boldsymbol\Phi(\omega^{N}) ) < \epsilon \right\}.
	\end{equation*}
\end{definition}

The empirical distribution method is a formalization of the simple heuristic of retaining a DGP if its average of sample marginals is close enough to the empirical distribution. Such a heuristic can be used when DGPs are i.i.d., for it happens to coincide with retaining DGPs that maximize the likelihood in this special case. With possible non-identical DGPs, while maximum likelihood is no longer useful, this heuristic remains valid. 

\begin{theorem}\label{thm_asymptotic}
	For all $\epsilon > 0$, the empirical distribution method accommodates the truth asymptotically almost surely. 
\end{theorem}

The proof of Theorem \ref{thm_asymptotic} is to verify that Kolmogorov's strong law of large numbers holds. This also suggests that the independence assumption is not crucial for this result. As long as there is a version of the strong law of large numbers, one can obtain the same conclusion. Notable cases include when the DGPs satisfy Markov property or are weakly dependent \citep{dejong_1995}. 

In addition, the empirical distribution method is not the only revision rule that accommodates the truth asymptotically almost surely. For instance, when the odd and even experiments are known to have different characteristics, one may apply the empirical distribution method separately to the odd and even experiments to obtain a potentially more refined data-revised set. Such revision rules, however, typically need to be tailored to the specific structure in a case-by-case manner. In contrast, the empirical distribution method stands out for its simplicity and general applicability. 

\subsection{Accommodating the Truth with Finite Sample}
As Theorem \ref{thm_asymptotic} holds for all $\epsilon > 0$, letting $\epsilon \rightarrow 0$ obtains the theoretic limit of the data-revised sets under the empirical distribution method. For applications with finite samples, the standard approach is to derive the $\epsilon$-bound as a function of the sample size that ensures a pre-specified asymptotic probability of accommodating the truth.\footnote{With a finite sample, the only revision rule that accommodates the truth almost surely is to keep using the initial set, because of the full-support assumption. Thus, the only meaningful notion of accommodating the truth with a finite sample is the probabilistic approach, which is also a standard practice in statistical inferences.} This section develops a novel and simple method to achieve this. First, define the following finite-sample notion of accommodating the truth.

\begin{definition}\label{def_confidence_level}
	A revision rule \textbf{accommodates the truth with an asymptotic level $1-\alpha$} if, for any initial set $\mathcal{P}$ and for every $P^{*} \in \mathcal{P}$, 
	\begin{equation*}
		\liminf_{N \rightarrow \infty} P^{*} (\{\omega^{N}: \mathcal{P}(\omega^{N}) \text{ accommodates } P^{*}\}) \geq 1-\alpha.
	\end{equation*}
\end{definition}

A revision rule satisfying Definition \ref{def_confidence_level} ensures that, regardless of which DGP governs the data, the data-revised set accommodates the true DGP with asymptotic probability at least $1-\alpha$. In particular, the level $1-\alpha$ is understood uniformly --- the probability bound holds uniformly over all possible DGPs in $\mathcal{P}$. Consequently, objective improvement is guaranteed with at least the same asymptotic probability whenever accommodating the truth is sufficient.

Definition \ref{def_confidence_level} effectively requires the data-revised sets to serve as consistent \textit{confidence regions} for the true DGP.\footnote{The coverage is in the weaker sense of accommodating but the difference is unimportant.} One way to ensure this property is to construct the confidence regions directly and use them as the data-revised sets. 

Constructing confidence regions is theoretically straightforward exploiting the well-known duality with hypothesis tests. For any $P \in \mathcal{P}$, consider testing the null hypothesis $P^{*} = P$, against the unrestricted alternative $P^{*} \neq P$. Let $A_{N, \alpha}(P) \subseteq S^{N}$ denote the \emph{region of acceptance}: the set of sample data under which the null cannot be rejected. Then find regions that satisfy 
\begin{equation*}
	\liminf_{N \rightarrow \infty} P(A_{\alpha,N}(P)) \geq 1-\alpha, 
\end{equation*}
that is, the probability of accepting the null when it is true is at least $1-\alpha$ asymptotically. Given any data $\omega^{N}$, construct the data-revised set as 
\begin{equation*}
	\mathcal{P}(\omega^{N}) = \{P \in \mathcal{P}: \omega^{N} \in A_{\alpha, N}(P)\}.
\end{equation*}
Such data-revised sets contain the true DGP with asymptotic probability $1-\alpha$, uniformly across all $P$ in $\mathcal{P}$.

When all possible DGPs are i.i.d., this construction is standard and tractable for two convenient features. First, by the central limit theorem, the regions of acceptance can be obtained from probability contours of the corresponding multivariate Gaussian approximations. The mean vectors and covariance matrices depend only on the unique marginal distribution and thus remain fixed as the sample size grows. Second, because each i.i.d. distribution is uniquely determined by its marginal, the number of tests remains fixed regardless of the sample size.  	

Both features do not carry over when DGPs may be non-identical. First, in this case, both the mean vectors and covariance matrices depend on all marginals across sample states. As a result, for every sample size, both need to be recomputed even for the same DGP. Second, a non-identical DGP is determined by all its marginals, so each additional observation increases the number of tests. 

To overcome these difficulties, I present in the following a novel method for constructing confidence regions for non-identical DGPs. This method retains the simplicity of the i.i.d. case while guaranteeing that the resulting confidence regions cover the true non-identical DGP with at least the required asymptotic level. 

Formally, for any $p \in \Delta(S)$, let $p^{\infty}$ denote the i.i.d. distribution over $\Omega$ with marginal $p$. Let $A^{*}_{N,\alpha} (p^{\infty})$ denote its region of acceptance with asymptotic level $1-\alpha$, constructed using the corresponding Gaussian approximation.\footnote{The exact form is standard and is given by equation \eqref{equ_iid_confidence_region} in the appendix, with additional notations.} For any $P \in \Delta_{indep}(\Omega)$, recall $\bar{P}^{N} \in \Delta(S)$ denotes the average of sample marginals. Let $A^{*}_{N, \alpha}((\bar{P}^{N})^{\infty})$ denote the region of acceptance constructed using the Gaussian approximation of the i.i.d. distribution $(\bar{P}^{N})^{\infty}$. Consider the following revision rule: 

\begin{definition}\label{def_robust_iid}
	The data-revised sets are obtained using the \textbf{augmented i.i.d. test} with asymptotic level $1-\alpha$ if, for every $\omega^{N}$, 
	\begin{equation*}
		\mathcal{P}(\omega^{N}) = \left\{P \in \mathcal{P}: \omega^{N} \in A^{*}_{N,\alpha}((\bar{P}^{N})^{\infty}) \right\}.
	\end{equation*}
\end{definition}

Intuitively, the augmented i.i.d. test follows a two-step procedure: 
\begin{enumerate}[(i)]
	\item For any sample data $\omega^{N}$, construct a confidence region as if the initial set consists of all i.i.d. DGPs. 
	\item For each DGP in the initial set, retain it in the data-revised set if its average of sample marginals coincides with the marginal of some i.i.d. DGP in the previous confidence region. 
\end{enumerate}

Notice the first step is the standard procedure for constructing confidence regions from i.i.d. sample data. The essential departure is the second step, which \textit{augments} the i.i.d. confidence region by also including the possible non-identical DGPs. Importantly, this augmentation is achieved through a straightforward comparison, adding virtually no computational difficulty. Therefore, implementing the augmented i.i.d. test is as tractable as conventional statistical inferences based on i.i.d. samples.

\begin{theorem}\label{thm_confident}
	The augmented i.i.d. test with asymptotic level $1-\alpha$ accommodates the truth with the same asymptotic level. 
\end{theorem}

The proof of this theorem relies on a key observation: For all $P \in \Delta_{indep}(\Omega)$, for all $N$ and  $\alpha$, 
\begin{equation*}
	A^{*}_{N, \alpha}(P) \subseteq A^{*}_{N,\alpha}((\bar{P}^{N})^{\infty}), 
\end{equation*}
which further implies that
\begin{equation*}
	P(A^{*}_{N, \alpha}(P))  \leq P(A^{*}_{N,\alpha}((\bar{P}^{N})^{\infty})).
\end{equation*}

Therefore, when testing the null hypothesis $P^{*} = P$, using the region of acceptance $A^{*}_{N,\alpha}((\bar{P}^{N})^{\infty})$, the probability of accepting $P$ when it is true is at least greater than the probability when using $A^{*}_{N,\alpha}(P)$. The latter probability is, by construction, asymptotically greater than $1-\alpha$. 

The key relation, $A^{*}_{N, \alpha}(P) \subseteq A^{*}_{N,\alpha}((\bar{P}^{N})^{\infty})$, is shown in Lemma \ref{lem2} by deriving a result relating the average covariance matrices of the two distributions, $P$ and $(\bar{P}^{N})^{\infty}$. Specifically, subtracting the average covariance matrix of $P_{N}$ from that of $((\bar{P}^{N})^{\infty})_{N}$ yields a positive semi-definite matrix. This result generalizes a well-known variance inequality for mixtures of binomial distributions \citep{wang1993} to the multinomial case.\footnote{Specifically, the average variance of an i.i.d. binomial distribution is always weakly greater than the average variance of a non-identical binomial distribution whose average mean is the same as the i.i.d. distribution.} 

\begin{remark} One potential concern is that computing the probability contours of multivariate Gaussian distributions may be difficult when $|S|$ is large. A practical alternative is to construct a Bonferroni-type confidence region by forming confidence intervals for the probability of each outcome $s$, each with confidence level $1-\alpha/(|S|-1)$. The intersection of all such confidence intervals then yields a confidence region with overall confidence level $1-\alpha$. However, this region is generally more conservative than that obtained directly from the multivariate Gaussian distribution. Moreover, using the corresponding result in \citet{wang1993} for binomial distributions, one can show that the Bonferroni-type confidence region constructed from i.i.d. distributions guarantees at least the same coverage probability for non-identical distributions.
\end{remark}

\section{Applications with Parametric Models}\label{application}
This section illustrates how the theoretical results translate into familiar statistical and economic settings. In particular, it studies two applications where the initial sets are given by specific parametric models. The first application showcases a setting where the data-revised sets obtained under the augmented i.i.d. test have a closed-form solution, and are closely related to the standard Wilson confidence interval. The second application highlights new findings in a commonly studied model of learning under ambiguity. These applications confirm the practical relevance of the proposed revision rules and the associated theoretical results.

\subsection{Bernoulli Model with Ambiguous Nuisance Parameters}\label{Bernoulli}
This model is a generalization of the one studied in \cite{Walley1991-WALSRW}. Suppose the DM faces a sequence of coin flips, with outcome space $\Omega = \{H, T\}^{\infty}$, representing heads and tails. The probability of getting a head from the $i$-th coin flip is determined by both a \textit{structural parameter} $\theta \in [0,1]$ and a \textit{nuisance parameter} $\psi_{i} \in [0,1]$: 
\begin{equation*}
	(1-\delta)\theta + \delta \psi_{i}  \in [0,1], 
\end{equation*}
for some fixed $\delta \in [0,1]$. The structural parameter is common across all flips and can be interpreted as a systematic characteristic of the coins. But each coin flip is also affected by its idiosyncratic feature captured by $\psi_{i}$. Throughout, I use the probability of $H$ to represent a probability distribution over $\{H, T\}$. Because each $\psi_{i}$ is only known to lie in $[0,1]$, each structural parameter $\theta$ corresponds to a set of possible DGPs: 
\begin{equation*}
	\mathcal{P}_{\theta} = \{P \in \Delta_{indep}(\Omega): P_{i} \in [(1-\delta)\theta, (1-\delta)\theta + \delta], \forall i\}.
\end{equation*}
Let $\Theta = [0,1]$ denote the set of structural parameters. The initial set of DGPs is therefore $\mathcal{P} = \cup_{\theta \in \Theta} \mathcal{P}_{\theta}$. The DM observes outcomes of $N$ coin flips. Their goal is to estimate the true structural parameter and make a set-valued prediction for the probability of getting a head in the next flip. The benchmark estimate is simply the initial set of parameters $\Theta = [0,1]$, so the benchmark prediction is $\mathcal{P}_{N+1} = [0,1]$. 

Consider the DM's asymptotic prediction using the empirical distribution method. For simplicity, I ignore the pre-specified $\epsilon$ by taking it to be arbitrarily small. Then the data-revised set is given by
\begin{equation*}
	\mathcal{P}(\omega^{N}) = \left\{P \in \mathcal{P}: N^{-1} \sum\limits_{i=1}^{N}P_{i} = \Phi(\omega^{N}) \right\}.
\end{equation*}
Let the DM's data-revised estimate of the structural parameter be given by
\begin{equation*}
	\Theta(\omega^{N}) \equiv \{\theta \in \Theta: \mathcal{P}_{\theta} \cap \mathcal{P}(\omega^{N}) \neq \emptyset \},
\end{equation*}
i.e., a structural parameter is considered possible whenever there is a corresponding DGP retained in $\mathcal{P}(\omega^{N})$. For any $\theta \in \Theta$, there exists $P \in \mathcal{P}_{\theta}$ satisfying the above condition if and only if $\Phi(\omega^{N}) \in \left[(1-\delta)\theta, (1-\delta)\theta + \delta \right]$.

As a result, it follows that
\begin{equation*}
	\Theta(\omega^{N})  = \left[ \max \left\{ \frac{\Phi(\omega^{N}) - \delta}{1-\delta} , 0 \right\} , \min \left\{ \frac{\Phi(\omega^{N}) }{1-\delta}, 1\right\} \right].
\end{equation*}
Notice the data-revised prediction is also completely determined by the data-revised estimate of the structural parameter and is given by
\begin{equation*}
	\mathcal{P}(\omega^{N})_{N+1} = \left[ \max \left\{ \Phi(\omega^{N}) - \delta , 0 \right\} , \min \left\{ \Phi(\omega^{N}) + \delta, 1\right\} \right].
\end{equation*}
Intuitively, the asymptotic prediction under the empirical distribution method is a ``$\delta$-fattening'' of the observed empirical frequency of heads. 

Next, consider the finite-sample estimate and prediction at asymptotic level $1-\alpha$ using the augmented i.i.d. test. For any sample data $\omega^{N}$, the first step is to construct the confidence interval as if the underlying DGPs were i.i.d. binomial distributions. Specifically, the corresponding confidence interval is the \textit{Wilson Interval}.\footnote{Different from the probably more famous Wald Interval which uses the sample variance, Wilson Interval is constructed by directly inverting the statistical tests, thus using the null variance. Wilson Interval has considerably better asymptotic performance than the Wald Interval. See \citet{Brown2001} for a discussion.} Let $z_{\alpha/2}$ denote the upper $100(\alpha/2) \%$ quantile of the standard normal distribution. Let $[\underline{W}(\omega^{N}), \overline{W}(\omega^{N})]$ denote the Wilson Interval which has the following closed-form expressions: 
\begin{align*}
	& \overline{W}(\omega^{N}) = \frac{N \Phi(\omega^{N}) + z_{\alpha/2}^{2}/2}{N + z_{\alpha/2}^{2}} + \frac{z_{\alpha/2} N^{1/2}}{N + z_{\alpha/2}^{2}} \left( \Phi(\omega^{N})(1- \Phi(\omega^{N})) + z_{\alpha/2}^{2}/(4N) \right)^{1/2};\\
	& \underline{W}(\omega^{N}) = \frac{N \Phi(\omega^{N}) + z_{\alpha/2}^{2}/2}{N + z_{\alpha/2}^{2}} - \frac{z_{\alpha/2} N^{1/2}}{N + z_{\alpha/2}^{2}} \left( \Phi(\omega^{N})(1- \Phi(\omega^{N})) + z_{\alpha/2}^{2}/(4N) \right)^{1/2}.
\end{align*}

Given the i.i.d. confidence interval, the second step is to consider non-identical DGPs whose average of sample marginals falls into this confidence interval: A structural parameter $\theta$ is retained in the data-revised estimate if and only if $[(1-\delta)\theta, (1-\delta)\theta + \delta] \cap [\underline{W}(\omega^{N}), \overline{W}(\omega^{N})] \neq \emptyset$. Hence, the data-revised estimate is
\begin{equation*}
	\Theta(\omega^{N})  = \left[ \max \left\{ \frac{\underline{W}(\omega^{N}) - \delta}{1-\delta} , 0 \right\} , \min \left\{ \frac{\overline{W}(\omega^{N})}{1-\delta}, 1\right\} \right].
\end{equation*}
Similarly, the data-revised prediction in this case is 
\begin{equation*}
	\mathcal{P}(\omega^{N})_{N+1} = \left[ \max \left\{ \underline{W}(\omega^{N}) - \delta , 0 \right\} , \min \left\{\overline{W}(\omega^{N}) + \delta, 1\right\} \right].
\end{equation*}
Notice the prediction is again a $\delta$-fattening of the Wilson Interval. Therefore, The resulting expressions are not only analytically tractable but also intuitively interpretable.

\subsection{Gaussian Signals with Ambiguous Variances}\label{Gaussian}
Prior-by-prior or full-Bayesian updating is the most commonly used update rule in models of learning under ambiguity in the literature. However, its asymptotic result is often hard to derive and is known only in specific parametric models. This section revisits one such model from \cite{reshidi2020information} and shows that applying the empirical distribution method yields a simpler analysis and entirely different conclusions.

In this model, the DM aims to learn the state of the world $\theta \in \Theta \equiv \R$ by observing a countably infinite sequence of signals denoted by $\{x_{i}\}_{i=1}^{\infty}$. Each $x_{i}$ is a Gaussian random variable with mean $\theta$ and variance $\sigma_{i}^{2}$, and let $g_{i}(\theta, \sigma_{i})$ denote its probability density function. The signals are mutually independent, but each $\sigma_{i}$ is only known to lie in $[\underline{\sigma}, \overline{\sigma}]$. Thus, the DM observes a sequence of independent but possibly heterogeneous Gaussian random variables. 

Formally, for each state $\theta$, let 
\begin{equation*}
	\mathcal{P}_{\theta} = \left\{\prod_{i=1}^{\infty} g_{i}(\theta, \sigma_{i}) : \sigma_{i} \in [\underline{\sigma}, \overline{\sigma}] , \forall i \right\}
\end{equation*}
denote the set of possible DGPs over the signal sequence. The initial set is $\mathcal{P} = \cup_{\theta \in \Theta} \mathcal{P}_{\theta}.$ For every $N \in \N$, let $\hat{x}^{N} \equiv (\hat{x}_{1}, \hat{x}_{2}, \cdots, \hat{x}_{N})$ be a sequence of signal realizations and let $\mathcal{P}(\hat{x}^{N})$ denote the data-revised set of DGPs. Because the DM's goal is to learn the true state, define
\begin{equation*}
	\Theta(\hat{x}^{N} )\equiv  \{\theta \in \Theta: \mathcal{P}_{\theta} \cap \mathcal{P}(\hat{x}^{N} ) \neq \emptyset \}
\end{equation*}
as the set of states consistent with the data-revised set of DGPs.

Directly applying full Bayesian updating here would amount to retain all possible DGPs, implying $\mathcal{P}(\hat{x}^{N} ) \equiv \mathcal{P}$ and hence $\Theta(\hat{x}^{N}) \equiv \Theta$ for all $\hat{x}^{N}$. In \citet{reshidi2020information}, they apply full Bayesian updating differently by assuming a prior $\mu$ over $\Theta$.  Then applying full Bayesian updating is to apply Bayes' rule to update $\mu$ under each possible DGP. Specifically, for each $\theta$, let $P_{\theta} \in \mathcal{P}_{\theta}$ denote a specific DGP, the posterior probability 
\begin{equation*}
	\mu(\theta | \hat{x}^{N}) = \frac{\mu(\theta) P_{\theta} (\hat{x}^{N})}{\int_{\Theta}\mu(\theta') P_{\theta'} (\hat{x}^{N})d\mu(\theta')}.
\end{equation*}	
This yields a set of posterior beliefs over $\Theta$.\footnote{This way of applying full Bayesian updating can be incorporated into the present framework by letting the state space be $\Theta \times S^{\infty}$ and allowing dependence in DGPs. Applying \eqref{equ_conditional} in Appendix \ref{apx: dependent} yields exactly the same posterior distribution over $\Theta$.} Their main result (Theorem 1) shows that in any state $\theta$ and some possible DGP, the set of posteriors converges almost surely to a set of degenerate distributions over a non-singleton set of states. That is, ambiguity does not vanish asymptotically.

Consider revising the initial set using the empirical distribution method. Because only the mean matters, it suffices to consider their sample mean. 
\begin{definition}
	The revision of states is obtained by the \textbf{sample mean method} if, for some pre-specified $\epsilon > 0$, and for every $\hat{x}^{N}$, 
	\begin{equation*}
		\Theta(\hat{x}^{N} ) = \left\{\theta \in \Theta: \left|N^{-1}\sum\limits_{i=1}^{N}\hat{x}_{i} - \theta \right| < \epsilon \right\}.
	\end{equation*}
\end{definition}
As in the empirical distribution method, $\Theta(\hat{x}^{N} )$ retains a state if it is close enough to the sample mean. Because the mean of each marginal distribution equal to $\theta$, applying Kolmogorov's strong law of large numbers yields the following result. 

\begin{proposition}\label{prop_gaussian}
	For any $\theta^{*} \in \Theta$ and any $P^{*} \in \mathcal{P}_{\theta^{*}}$, the revision of states obtained by the sample mean method with any $\epsilon > 0$ contains the true state asymptotically almost surely, i.e., for any $\epsilon > 0$, 
	\begin{equation*}
		P^{*}\left(\hat{x}: \lim\limits_{N \rightarrow \infty} \theta^{*} \in \Theta(\hat{x}^{N})\right) = 1.
	\end{equation*}
\end{proposition}

As the conclusion holds for any $\epsilon > 0$, taking $\epsilon$ arbitrarily small makes the revised set of states arbitrarily precise. Thus, even with Gaussian signals that have unknown and possibly heterogeneous variances, the true state can still be identified asymptotically almost surely. This stands in sharp contrast to full Bayesian updating, under which ambiguity persists. 

The key factor enabling asymptotic identification here is that all Gaussian signals share the same mean. When the means themselves are also ambiguous, the sample-mean method may still yield asymptotic ambiguity over a non-vanishing set of states. Nonetheless, the key here is that this asymptotic prediction is obtained straightforwardly with the sample-mean method, whereas the corresponding asymptotic result for full-Bayesian updating in this setting remains unknown.

\section{Related Literature}\label{literature}
The decision environment formulated in this paper is closely related to some in the literature on decisions under ambiguity. In particular, it directly generalizes the setting studied by \cite{Epstein2007}. They assume the DM applies maximum likelihood updating to revise the initial set. The present paper highlights possible concerns with this approach. \citet{Epstein2016} develop robust confidence regions when the possible data-generating processes are belief functions. Belief functions impose restrictions on the possible marginal distributions. In contrast, the environment studied here allows for arbitrary marginals. But the main conceptual difference from their paper and other papers on asymptotic learning under ambiguity, such as \cite{Marinacci2002} and \cite{MARINACCI2019144}, is that the present paper emphasizes implications for decision making in addition to asymptotic learning.

This paper also contributes to the literature on dynamic decisions under ambiguity by proposing new rules for revising or updating sets of distributions. See \citet{gilboa_marinacci_2013} and \citet{CHENG2022102587} for recent developments. The essential departure of the present paper from this literature is that it evaluates decisions using an objective criterion. The objective criterion leads to a characterization of accommodating-the-truth property, conceptually analogous to statistical consistency. In this way, this paper draws a connection between a classical concept from statistics and the theory of decisions under ambiguity, following the line of research by \citet{cerreia-vioglio2013} and \cite{denti2022}. This objective approach also resonates with some recent studies of misspecified learning that evaluate performance according to an objective measure, such as \citet{Frick2021} and \citet{he2020evolutionarily}.  

Finally, this paper develops a useful augmenting technique for making inferences in the presence of independent but non-identical distributions. In essence, this technique can be applied on top of standard statistical procedures developed under the i.i.d. assumption. The data-revised set naturally serves as a set-valued identification object and is therefore related to the partial identification literature (see surveys by  \citet{Tamer2010}, \citet{canay/shaikh:2017}, and \citet{MOLINARI2020355}). In that literature, the decision environment is typically formulated so that the data-generating distribution is point identified, while the payoff-relevant parameters are only partially identified, see, for example, \citet{christensen2023optimal}. By contrast, the present paper focuses on situations where the distribution itself is partially identified and develops relevant inference methods for such settings. 

\appendix
\numberwithin{equation}{section}
\numberwithin{definition}{section}
\numberwithin{theorem}{section}
\numberwithin{corollary}{section}
\numberwithin{proposition}{section}
\numberwithin{remark}{section}

\section{Proofs of Results}\label{sec: main_proofs}

For ease of exposition, unless otherwise specified, I use the notations $\mathcal{P}$, $\mathcal{P}(\omega^{N})$, and $P$ to represent $\overline{co}(\mathcal{P}){N}^{\infty}$, $\overline{co}(\mathcal{P}(\omega^{N}){N}^{\infty})$, and $P_{N}^{\infty}$, respectively. Accordingly, a truth-accommodating refinement is one that satisfies $P^{*}\in\mathcal{P}(\omega^{N})\subsetneqq\mathcal{P}$.

\subsection{Proof of Theorem \ref{thm: truth_accommodating_refinement_basic}}
\begin{proof}
Suppose $\mathcal{P}(\omega^{N})$ does not accommodate $P^{*}$, i.e., $P^{*} \notin \mathcal{P}(\omega^{N})$. By the strong separating hyperplane theorem (See, for example, Corollary 5.80 in \cite{aliprantis2006}), there exists a $\Sigma_{N}^{\infty}$-measurable bounded function $f: \Omega_{N} \rightarrow \R$ and a real number $x \in \R$ such that
	\begin{equation*}
	\int_{\Omega_{N}} f(\tilde{\omega}_{N}) dP^{*}(\tilde{\omega}_{N}) < x, \text{ and } 
	\min\limits_{P \in \mathcal{P}(\omega^{N})} \int_{\Omega_{N}} f(\tilde{\omega}_{N}) dP(\tilde{\omega}_{N}) > x.  
	\end{equation*}
    Consider the basic decision problem $D = \{f, x\}$, the above inequalities imply that $c(D) = x$, $c(D, \omega^{N}) = f$, but $W(f, P^{*}) < x = W(x, P^{*})$. 
\end{proof}

\subsection{Proof of Theorem \ref{thm_monotone}}
\begin{proof}
    The necessity is a direct consequence of Theorem \ref{thm: truth_accommodating_refinement_basic} as basic decision problems are monotone. The sufficiency is proved as follows.

    Fix any initial set $\mathcal{P}$. By definition, there exists $g \in \mathcal{F}$ such that for all $f \in D$, $f = \lambda_{f} g + c_{f}$ for some $\lambda_{f} \geq 0$ and $c_{f} \in \R$. Since for any $f \in D$ and any $P \in \Delta(\Omega)$,
    \begin{equation*}
        W(f,P) = \lambda_{f} W(g,P) + c_{f},
    \end{equation*}
	it follows that all $f \in D$ and $g$ share the same worst-case DGP that delivers the lowest $W(g,P)$. Based on this observation, further denote
    \begin{align*}
        &[a, b] \equiv \left[\min_{P \in \mathcal{P}} W(g,P), \max_{P \in \mathcal{P}} W(g,P) \right] \subseteq \R, \\
        &[a', b'] \equiv \left [\min_{P \in \mathcal{P}(\omega^{N}) } W(g,P), \max_{P \in \mathcal{P}(\omega^{N})} W(g,P) \right ]\subseteq \R,\\
        &p^{*} \equiv W(g, P^{*}) \in \R.
    \end{align*}  
    Then $\mathcal{P}(\omega^{N})$ is a truth-accommodating refinement of $\mathcal{P}$ implies that $p^{*} \in [a',b'] \subseteq [a,b]$. 

    Denote $c(D) = f_{1}$ and $c(D, \omega^{N}) = f_{2}$. If $f_{1} = f_{2}$, then the conclusion holds trivially. Otherwise, it must hold that 
    \begin{align*}
        &\lambda_{f_{1}} \cdot a + c_{f_{1}} \geq \lambda_{f_{2}} \cdot a + c_{f_{2}}, \\
        &\lambda_{f_{1}} \cdot a' + c_{f_{1}} < \lambda_{f_{2}} \cdot a' + c_{f_{2}}.
    \end{align*} 
    The second inequality must be strict, otherwise, since $f_{1}$ is a solution under the revised set, the tie-breaking assumption forces $c(D, \omega^{N}) = f_{1}$, contradicting $c(D, \omega^{N}) = f_{2} \neq f_{1}$. Strict inequality further implies that $a' > a$. Then, as $\lambda_{f_{1}}$ and $\lambda_{f_{2}}$ are both non-negative, combining both inequalities yields $\lambda_{f_{2}} > \lambda_{f_{1}}$. Thus, for any $p \geq a'$, it must hold that 
    \begin{equation*}
        \lambda_{f_{1}} \cdot p + c_{f_{1}} < \lambda_{f_{2}} \cdot p + c_{f_{2}}.
    \end{equation*}
    Therefore, as $p^{*} \geq a'$, it follows that $W(c(D, \omega^{N}), P^{*}) > W(c(D), P^{*})$. 
	
	To show the existence of a monotone decision problem where strict improvement occurs, one can apply a separating hyperplane argument to construct a basic decision problem with $f_{1} \neq f_{2}$.
    \end{proof}

\subsection{Proof of Theorem \ref{thm_monotone_converse}}
\begin{proof}
   \textbf{(i)} This follows from Theorem \ref{thm_impossible}. If $\mathcal{P}(\omega^{N})$ refines $\mathcal{P}$, there always exists $P^{*} \in \mathcal{P}(\omega^{N})$ (for example, those on the boundary) such that the condition in Theorem \ref{thm_impossible} fails. Then by Theorem \ref{thm_monotone}, the decision problem where objective improvement fails must be a non-monotone decision problem. 
   \bigskip 

   \textbf{(ii)} Let $D = \{f_{1}, f_{2}\} \in \mathcal{D}$ be a non-monotone decision problem with finitely-based acts. For some finite $M > N$, let $S_{N}^{M}$ denote the finite set of experiments where both acts are measurable and let $\omega_{N}^{M} = (\omega_{N+1}, \cdots, \omega_{M})$. For each $P \in \Delta_{indep}(S_{N}^{M})$, identify it as a vector $(P_{N+1}, \cdots, P_{M}) \in (\Delta(S))^{M-N} \subset \R^{(M - N) \times |S|}$. Notice the set $(\Delta(S))^{M-N}$ is convex and compact.\footnote{A convex combination in the space $(\Delta(S))^{M-N}$ is different from that in $\Delta_{indep}(S_{N}^{M})$. The former is a convex combination of the marginal distributions, while the latter is a convex combination of the joint distributions. Thus, this convexity does not contradict the fact that $\Delta_{indep}(S_{N}^{M})$ is not convex.} Next, for $j \in \{1, 2\}$, the function
   \begin{equation*}
    W(f_{j}, P) = \sum_{\omega_{N}^{M} \in S_{N}^{M}} \prod_{i=N+1}^{M} P_{i}(\omega_{i}) f_{j}(\omega_{N}^{M}),
   \end{equation*}
   is continuous over $(\Delta(S))^{M-N}$. Therefore, for each $x \in (\min_{\omega_{N}^{M}} f_{j}(\omega_{N}^{M}), \max_{\omega_{N}^{M}} f_{j}(\omega_{N}^{M}))$, there exists a full support $P \in \Delta_{indep}(S_{N}^{M})$ such that $W(f_{j}, P) = x$.

   $D$ is non-monotone implies both $f_{1}$ and $f_{2}$ are non-constant. Pick any $\hat{\omega}_{N} \in \Omega_{N}$ and define 
   \begin{equation*}
	 \hat{f}_{j}(\tilde{\omega}_{N}) = f_{j}(\tilde{\omega}_{N}) - f_{j}(\hat{\omega}_{N}), \text{for } j \in \{1,2\}. 
	\end{equation*}
	Observe that $D$ is monotone if and only if $\hat{f}_{1}$ and $\hat{f}_{2}$ are in the same ray, i.e., there exists $\alpha > 0$ such that $\hat{f}_{2} = \alpha \hat{f}_{1}$. Thus, when $D$ is non-monotone, there are two possible cases:
    \bigskip 

	\textbf{Case 1:} $\hat{f}_{1}$ and $\hat{f}_{2}$ are in the opposite ray, i.e., there exists $\alpha < 0$ such that $\hat{f}_{2} = \alpha \hat{f}_{1}$. 

	Since $S^{M}_{N}$ is finite, all relevant acts and DGPs can be identified as vectors in a finite-dimensional Euclidean space. Henceforth, I use the standard inner product $\langle P, f \rangle$ to denote the expected payoff $W(f, P)$. First observe that, fix an $f$, any $P \in \Delta(S_{N}^{M})$ can be uniquely decomposed into $\hat{P}$ and $\hat{P}^{\perp}$ such that 
	\begin{equation*}
			\langle P, \hat{f} \rangle = \langle \hat{P}, \hat{f} \rangle + \langle \hat{P}^{\perp}, \hat{f} \rangle = \langle \hat{P}, \hat{f} \rangle + 0. 
	\end{equation*}
	That is, $\hat{P}$ is parallel to $\hat{f}$ and $\hat{P}^{\perp}$ is orthogonal to $\hat{f}$. When $\hat{f}_{2}$ and $\hat{f}_{1}$ are in the opposite ray, it follows that the same decomposition works for both $\hat{f}_{1}$ and $\hat{f}_{2}$. Thus, for any $P \in \Delta_{indep}(S_{N}^{M})$, let $\hat{P}$ denote its component parallel to both $\hat{f}_{1}$ and $\hat{f}_{2}$, then 
	\begin{equation*}
		\langle P, \hat{f}_{j} \rangle = \langle \hat{P}, \hat{f}_{j} \rangle, \text{for } j \in \{1,2\}. 
	\end{equation*}
	And thus,
	\begin{equation*}
		\langle P, f_{j} \rangle = f_{j}(\hat{\omega}_{N}) + \text{sgn}(\langle \hat{P}, \hat{f}_{j} \rangle)  \Vert \hat{P} \Vert\cdot \Vert \hat{f}_{j} \Vert, \text{for } j \in \{1,2\},
	\end{equation*}
	where $ \text{sgn}(\langle \hat{P}, \hat{f}_{j} \rangle) $ is always opposite for $f_{1}$ and $f_{2}$. 
	
    Fix any initial set $\mathcal{P}$ that includes both $P$ and $P'$. Define the following set using $f_{1}$: 
	\begin{equation*}
		\{\text{sgn}(\langle \hat{P}, \hat{f}_{1} \rangle)  \Vert \hat{P} \Vert: P \in \overline{co}(\mathcal{P})\},
	\end{equation*}
	which is closed and convex, and thus can be represented by a closed interval $[a,b] \subseteq \R$. Then for any $P \in \mathcal{P}$, let $p \in [a,b]$ denote its corresponding element in $[a,b]$ and notice
	\begin{align*}
		& \langle P, f_{1} \rangle = f_{1}(\hat{\omega}_{N}) + \Vert \hat{f}_{1} \Vert p,\\
		& \langle P, f_{2} \rangle = f_{2}(\hat{\omega}_{N}) - \Vert \hat{f}_{2} \Vert p. 
	\end{align*}
	It further implies that $a$ is the minimizer for $f_{1}$ and $b$ is the minimizer for $f_{2}$. Let $P_{1}$ and $P_{2}$ be the two minimizers in $\mathcal{P}$ corresponding to $a$ and $b$, respectively. Suppose $c(D) = f_{1}$, the other case is symmetric. 

    As $P, P' \in \mathcal{P}$, both acts can be the unique maximizer of $W(f,P)$ among $P \in \mathcal{P}$, thus 
	\begin{align*}
		&f_{1}(\hat{\omega}_{N}) + \Vert \hat{f}_{1} \Vert a < f_{2}(\hat{\omega}_{N}) - \Vert \hat{f}_{2} \Vert a,\\
		&f_{1}(\hat{\omega}_{N}) + \Vert \hat{f}_{1} \Vert b > f_{2}(\hat{\omega}_{N}) - \Vert \hat{f}_{2} \Vert b.
	\end{align*}
	Therefore, there then exists $c \in (a,b)$ such that  
	\begin{equation*}
		f_{1}(\hat{\omega}_{N}) + \Vert \hat{f}_{1} \Vert c = f_{2}(\hat{\omega}_{N}) - \Vert \hat{f}_{2} \Vert c.
	\end{equation*}
	Since $f_{1}(\hat{\omega}_{N}) + \Vert \hat{f}_{1} \Vert c > f_{1}(\hat{\omega}_{N}) + \Vert \hat{f}_{1} \Vert a$, there further exists $d \in (c, b)$ such that 
	\begin{equation*}
		f_{1}(\hat{\omega}_{N}) + \Vert \hat{f}_{1} \Vert a < f_{2}(\hat{\omega}_{N}) - \Vert \hat{f}_{2} \Vert d. 
	\end{equation*}
	By continuity of $\langle P, f_{2} \rangle$ in $P$, there exists $\tilde{P} \in \Delta_{indep}(S_{N}^{M})$ such that $\langle \tilde{P}, f_{2} \rangle - f_{2}(\hat{\omega}_{N}) = -\Vert \hat{f}_{2} \Vert d$, i.e., the data revised set $\mathcal{P}(\omega^{N}) = \{P_{1}, \tilde{P} \}$, after taking the closed and convex hull, would correspond to $[a, d] \subsetneqq [a,b]$. The last inequality implies $c(D, \omega^{N}) = f_{2}$. 
    
    Again by continuity, there exists $P^{*}\in \Delta_{indep}(S_{N}^{M})$ that corresponds to some $p^{*} \in (c,d]$. Further let the data-revised set be $\mathcal{P}(\omega^{N}) = \{P_{1}, \tilde{P}, P^{*} \}$, notice this does not change the data-revised decision and accommodates $P^{*}$. However, 
	\begin{equation*}
		f_{1}(\hat{\omega}_{N}) + \Vert \hat{f}_{1} \Vert p^{*} > f_{2}(\hat{\omega}_{N}) - \Vert \hat{f}_{2} \Vert p^{*},
	\end{equation*}
	i.e., $W(c(D), P^{*}) > W(c(D, \omega^{N}), P^{*})$.
    \bigskip 
	
	\textbf{Case 2:} $\hat{f}_{2}$ is not in the affine hull of $\hat{f}_{1}$. As $W(f_{1}, P) > W(f_{2}, P)$ and $W(f_{1}, P') < W(f_{2}, P')$, by continuity, there exists $P_{0} \in \Delta_{indep}(S_{N}^{M})$ such that $W(f_{1}, P_{0}) = W(f_{2}, P_{0})$. 
    
    For any $P \in \Delta_{indep}(S_{N}^{M})$, decompose it into $\hat{P}$ and $\hat{P}^{\perp}$ with respect to $\hat{f}_{2}$ so that 
	\begin{equation*}
		\langle P, \hat{f}_{2} \rangle = \langle \hat{P}, \hat{f}_{2} \rangle. 
	\end{equation*}
    Then for $f_{1}$, one has 
	\begin{equation*}
		\langle P, \hat{f}_{1} \rangle = \langle \hat{P}, \hat{f}_{1} \rangle + \langle \hat{P}^{\perp}, \hat{f}_{1} \rangle,
	\end{equation*}
	where $\langle \hat{P}^{\perp}, f_{1} \rangle \neq 0$ whenever $ \hat{P}^{\perp}$ is non-zero. 
    
    Since $P_{0}$ has full support, it lies in the interior of $(\Delta_{indep}(S_{N}^{M}))$ (identified as the set $(\Delta(S))^{M-N}$). Thus there always exists $\epsilon > 0$ such that the open ball (in $\R^{(M - N) \times |S|}$) centered at $P_{0}$ with radius $\epsilon$ is entirely contained in $\Delta_{indep}(S_{N}^{M})$. Let $\mathcal{P}$ be a superset of this open ball and suppose $c(D) = f_{1}$, the other case is symmetric. 
    
    From the open ball, one can also find $P_{\epsilon}$ such that 
	\begin{equation*}
		\hat{P}_{\epsilon} = \hat{P}_{0}, \text{and } \langle \hat{P}_{\epsilon}^{\perp}, \hat{f}_{1} \rangle > \langle  \hat{P}_{0}^{\perp}, \hat{f}_{1} \rangle.
	\end{equation*}
	This further implies 
	\begin{align*}
		\langle P_{\epsilon}, f_{1} \rangle & = f_{1}(\hat{\omega}_{N}) + \langle \hat{P}_{\epsilon}, \hat{f}_{1} \rangle + \langle \hat{P}_{\epsilon}^{\perp}, \hat{f}_{1} \rangle \\
		& =  f_{1}(\hat{\omega}_{N}) + \langle \hat{P}_{0}, \hat{f}_{1} \rangle + \langle \hat{P}_{\epsilon}^{\perp}, \hat{f}_{1} \rangle\\
		& >  f_{1}(\hat{\omega}_{N}) + \langle \hat{P}_{0}, \hat{f}_{1} \rangle + \langle \hat{P}_{0}^{\perp}, \hat{f}_{1} \rangle \\
		& = \langle P_{0}, f_{1} \rangle = \langle P_{0}, f_{2} \rangle = f_{1}(\hat{\omega}_{N}) + \langle \hat{P}_{0}, \hat{f}_{2} \rangle\\
		& = f_{2}(\hat{\omega}_{N}) + \langle \hat{P}_{\epsilon}, \hat{f}_{2} \rangle = \langle P_{\epsilon}, f_{2} \rangle. 
	\end{align*}
	Similarly, one can find $P_{-\epsilon}$ in the open ball with the property that 
	\begin{equation*}
		\langle P_{-\epsilon}, f_{1} \rangle <  \langle P_{0}, f_{1} \rangle = \langle P_{0}, f_{2} \rangle = \langle P_{-\epsilon}, f_{2} \rangle. 
	\end{equation*}
	Then let the data-revised set be $\{P_{-\epsilon}, P_{\epsilon} \}$ and let $P^{*} = P_{\epsilon}$. In this case, the data-revised decision is $c(D, \omega^{N}) = f_{2}$, but 
	\begin{equation*}
		W(f_{2}, P^{*}) =  \langle P_{\epsilon}, f_{2} \rangle < \langle P_{\epsilon}, f_{1} \rangle = W(f_{1}, P^{*}),
	\end{equation*}	
	i.e., the conclusion. 

    The following example provides a simpler illustration of the main intuition in Case 2: 
    \begin{example}\label{example: non_monotone}
    Suppose $D = \{f_{1}, f_{2}\}$ and all acts depend on a single experiment with three possible outcomes, i.e., $S = \{s_{1}, s_{2}, s_{3}\}$. Let $f_{1} = (0, 2, 0)$ and $f_{2} = (-1, 1, 1)$. $\hat{f}_{2}$ is not in the affine hull of $\hat{f}_{1}$. In terms of the marginal distribution over $S$, let $\mathcal{P} = \Delta(S)$, $\mathcal{P}(\omega^{N}) = \Delta(\{s_{2}, s_{3}\})$, and $P^{*} = \delta_{s_{2}}$. Notice $\mathcal{P}(\omega^{N})$ is a truth-accommodating refinement. Then $c(D) = f_{1}$, $c(D, \omega^{N}) = f_{2}$, but $W(f_{1}, P^{*}) = 2 > 1 = W(f_{2}, P^{*})$.
    \end{example}

    \begin{remark}\label{remark: non_finitely_based}
        The assumption of finitely-based acts is used to embed $\Delta_{indep}(S_{N}^{M})$ into Euclidean space and show that $W(f, P)$ is continuous over $\Delta_{indep}(S_{N}^{M})$. This continuity breaks down when the acts are not finitely based, i.e., depend on tail events. By Kolmogorov's zero-one law, all tail events have a probability of either 0 or 1, and thus $W(f, P)$ is no longer continuous in $P$. 
        
        In particular, consider a non-monotone decision problem where the two acts are bets on and against a tail event $A$. This corresponds to Case 1 in the proof. For this decision problem, however, all DGPs in $\mathcal{P}$ correspond to the two endpoints $a$ and $b$ and nothing else. Therefore, if the data-revised set is a truth-accommodating refinement, then the data-revised decision is either the same as the benchmark decision or exactly optimal under the true DGP. In other words, a truth-accommodating refinement, ``inconveniently", leads to an objective improvement in such a non-monotone decision problem. This fact, however, does not generalize to all non-monotone problems with tail-event acts. Notice the previous example is still true when replacing $s_{1}, s_{2}$ and $s_{3}$ by some tail events. The assumption of finitely-based acts is needed to simplify the statement in light of this subtlety.
    \end{remark}
\end{proof}

\subsection{Proof of Theorem \ref{thm_impossible}}\label{thm_impossible_proof}
\begin{proof}

	\textbf{IF.} The inequality $W(c(D, \omega^{N}), P^{*}) \geq W(c(D), P^{*})$ is trivially true when the data-revised and benchmark decisions are the same. Consider any $D\in \mathcal{D}$ where these two are different, denote $f = c(D)$ and $f' = c(D, \omega^{N})$. The condition in the theorem says that, for some $\alpha \in (0, 1]$, 
	\begin{equation}\label{equ_alpha}
		\mathcal{P}(\omega^{N}) = \alpha P^{*} + (1-\alpha)\mathcal{P}. 
	\end{equation}
	By optimality of $f$ and $f'$ under $\mathcal{P}$ and $\mathcal{P}(\omega^{N})$ respectively, one has 
	\begin{equation}\label{objectivedominanceinequality1}
		\min\limits_{P \in \mathcal{P}} \int_{\Omega_{N}}   f(\tilde{\omega}_{N})dP(\tilde{\omega}_{N})  \geq \min\limits_{P \in \mathcal{P}} \int_{\Omega_{N}}   f'(\tilde{\omega}_{N})dP(\tilde{\omega}_{N}) , 
	\end{equation}
	and 
	\begin{equation}\label{objectivedominanceinequality2}
		\min\limits_{P \in \mathcal{P}(\omega^{N})} \int_{\Omega_{N}} f'(\tilde{\omega}_{N})dP(\tilde{\omega}_{N}) > \min\limits_{P \in \mathcal{P}(\omega^{N})} \int_{\Omega_{N}} f(\tilde{\omega}_{N})dP(\tilde{\omega}_{N}),
	\end{equation}
	where the strict inequality is required by the tie-breaking assumption (Assumption \ref{assumption4}) and that $f \neq f'$. Substituting \eqref{equ_alpha} into \eqref{objectivedominanceinequality2} yields, 
	\begin{align*}
		&\alpha \int_{\Omega_{N}} f'(\tilde{\omega}_{N})dP^{*}(\tilde{\omega}_{N}) + (1-\alpha)\min\limits_{P \in \mathcal{P}} \int_{\Omega_{N}} f'(\tilde{\omega}_{N})dP(\tilde{\omega}_{N}) 	\\
		>& \alpha \int_{\Omega_{N}} f(\tilde{\omega}_{N})dP^{*}(\tilde{\omega}_{N}) + (1-\alpha)\min\limits_{P \in \mathcal{P}_{N}} \int_{\Omega_{N}} f(\tilde{\omega}_{N})dP(\tilde{\omega}_{N}). 
	\end{align*}
	Rearranging terms to get, 
	\begin{align*}
		&\alpha \left[ \int_{\Omega_{N}} f'(\tilde{\omega}_{N})dP^{*}(\tilde{\omega}_{N})  - \int_{\Omega_{N}} f(\tilde{\omega}_{N})dP^{*}(\tilde{\omega}_{N}) \right] \\
		>& (1-\alpha) \left[\min\limits_{P \in \mathcal{P}_{N}^{\infty}} \int_{\Omega_{N}} f(\tilde{\omega}_{N})dP(\tilde{\omega}_{N}) -\min\limits_{P \in \mathcal{P}} \int_{\Omega_{N}} f'(\tilde{\omega}_{N})dP(\tilde{\omega}_{N})  \right] \geq 0,
	\end{align*}
	where the last inequality follows from inequality \eqref{objectivedominanceinequality1}. As $\alpha > 0$, one has, 
	\begin{equation*}
		W(c(D, \omega^{N}),P^{*}) - W(c(D),P^{*}) > 0.
	\end{equation*}
	
    \bigskip 

	\textbf{ONLY IF.} Consider the contra-positive statement: If there does not exist any $\alpha \in (0,1]$ such that the condition holds, then there must exist $D \in \mathcal{D}$ with $W(c(D, \omega^{N}), P^{*}) < W(c(D), P^{*})$. If $\mathcal{P}(\omega^{N})$ does not accommodate $P^{*}$, then the conclusion follows from Theorem \ref{thm: truth_accommodating_refinement_basic}. Consider the case where $\mathcal{P}(\omega^{N})$ accommodates $P^{*}$.
	The presumption implies  
	\begin{equation*}
		\mathcal{P}(\omega^{N}) \neq \mathcal{P}_{\alpha} \equiv \alpha P^{*} + (1-\alpha)\mathcal{P}, \forall \alpha \in (0,1], 
	\end{equation*}
	It thus implies the existence of $\hat{\alpha} \in [0,1]$ such that 
	\begin{equation*}
		\mathcal{P}(\omega^{N}) \subsetneqq \mathcal{P}_{\alpha}, \forall \alpha \leq \hat{\alpha},
	\end{equation*}
	and for all $\alpha > \hat{\alpha}$, there always exists $P' \in \mathcal{P}(\omega^{N})$ such that $	P' \notin \mathcal{P}_{\alpha}$. In other words, $P'$ is on the boundary of both $\mathcal{P}(\omega^{N})$ and $\mathcal{P}_{\hat{\alpha}}$. Since both sets are closed and convex, they admit the same supporting hyperplane at $P'$. In other words, there exists an act $f$, whose minimum expectation among both $\mathcal{P}(\omega^{N})$ and $\mathcal{P}_{\hat{\alpha}}$ are achieved at $P'$. It therefore implies that, 
	\begin{equation*}
		\min\limits_{P \in \mathcal{P}(\omega^{N})} W(f, P)= \min\limits_{P \in \hat{\alpha} P^{*} + (1-\hat{\alpha})\mathcal{P}} W(f,P) = \hat{\alpha}W(f, P^{*}) + (1-\hat{\alpha}) \min\limits_{P \in \mathcal{P}} W(f,P).
	\end{equation*}
	
	Next, because $\mathcal{P}(\omega^{N}) \subsetneqq \mathcal{P}_{\hat{\alpha}}$, by strong separating hyperplane theorem, there also must exist an act $g$ such that for some $\beta > \hat{\alpha}$, 
	\begin{equation*}
		\hat{\alpha} W(g, P^{*}) + (1-\hat{\alpha}) \min\limits_{P \in \mathcal{P}} W(g, P) < \min\limits_{P \in \mathcal{P}(\omega^{N})} W(g, P) 
		= \beta  W(g, P^{*}) + (1-\beta) \min\limits_{P \in \mathcal{P}} W(g, P), 
	\end{equation*}
	where $\beta >\hat{\alpha}$ follows from the fact that for such $g$, $W(g, P^{*}) > \min\limits_{P \in \mathcal{P}} W(g,P)$. One can normalize $g$ by taking mixtures with constant acts to get 
	\begin{align*}
		W(g, P^{*}) = 1/2, \text{ and }  \min\limits_{P \in \mathcal{P}}  W(g, P) = 0. 
	\end{align*}
	Then normalize $f$ similarly to get for some $\epsilon_{1}, \epsilon_{2} \in (0, (\beta - \hat{\alpha})/2)$ such that 
	\begin{align*}
		W(f, P^{*}) = 1/2  + \epsilon_{1} \text{ and }  \min\limits_{P \in \mathcal{P}}  W(f, P) = \epsilon_{2}. 
	\end{align*}
	Then for the decision problem $D = \{f, g\}$, it is the case that $c(D) = f$ and $c(D, \omega^{N}) = g$. The second claim follows from
	\begin{align*}
		\min\limits_{P \in \mathcal{P}(\omega^{N})} W(f, P) &  = \hat{\alpha} W(f, P^{*}) + (1-\hat{\alpha}) \min\limits_{P \in \mathcal{P}}W(f, P)  \\
		& = \hat{\alpha}(1/2 + \epsilon_{1}) + (1-\hat{\alpha})\epsilon_{2}  < \beta/2 \\
		& = \beta W(g, P^{*}) + (1-\beta) \min\limits_{P \in \mathcal{P}}W(g, P) = \min\limits_{P \in \mathcal{P}(\omega^{N})} W(g, P).
	\end{align*}
	But $W(f, P^{*}) > W(g, P^{*})$, i.e. the conclusion. 
\end{proof}

\subsection{Proof of Proposition \ref{prop: improving_worst_case}}
\begin{proof}
    For any $D \in \mathcal{D}$, if the data-revised set accommodates the true DGP $P^{*}$ and refines the initial set, then 
	\begin{align*}
		W(c(D, \omega^{N}), P^{*}) &\geq \min\limits_{P \in \mathcal{P}(\omega^{N})} W(c(D, \omega^{N}), P) \\
		& \geq \min\limits_{P \in \mathcal{P}(\omega^{N})} W(c(D), P)\\
		& \geq \min\limits_{P \in \mathcal{P}} W(c(D), P),
	\end{align*}
	where the first inequality follows from $P^{*} \in \mathcal{P}(\omega^{N})$, the second inequality follows from the fact that $c(D, \omega^{N})$ is optimal with respect to $\mathcal{P}(\omega^{N})$, the third inequality follows from $\mathcal{P}(\omega^{N})\subsetneqq \mathcal{P}$. Notice the last inequality can be strict for some $D \in \mathcal{D}$ as $\mathcal{P}(\omega^{N}) \subsetneqq \mathcal{P}$.
\end{proof}

\subsection{Proof of Theorem \ref{thm_asymptotic}}\label{SLLN}
\begin{proof}
	The only argument is to show all possible DGPs satisfy the condition of Kolmogorov's strong law of large numbers:  
	
	\textbf{Kolmogorov's SLLN.} Let $\{X_{i}\}$ be a sequence of independent random variables. Define $Y_{N} = N^{-1}\sum_{i=1}^{N}X_{i}$ and $\bar{\mu}_{N} = N^{-1}\sum_{i=1}^{N} \mu_{i}$. If $E[X_{i}] = \mu_{i}$ and $var(X_{i}) = \sigma_{i}^{2}$, 
	\begin{equation*}
		\lim_{N \rightarrow \infty} \sum\limits_{i=1}^{N} \frac{\sigma_{i}^{2}}{i^{2}} < \infty,
	\end{equation*}
	then $Y_{N} - \bar{\mu}_{N} \rightarrow 0$ almost surely as $N \rightarrow \infty$. 
	
	Fix any possible DGP $P \in \Delta_{indep}(\Omega)$. For any $s \in S$, let $X_{i} = I\{\omega_{i} = s\}$. Then $X_{i}$'s are independent random variables with $E[X_{i}] = P_{i}(s)$ and $var(X_{i}) = P_{i}(s)(1-P_{i}(s))$. To see the condition of SLLN holds, 
	\begin{align*}
		\lim_{N \rightarrow \infty} \sum\limits_{i=1}^{N} \frac{\sigma_{i}^{2}}{i^{2}} & =  \lim_{N \rightarrow \infty} \sum\limits_{i=1}^{N} \frac{P_{i}(s_{j})(1-P_{i}(s_{j}))}{i^{2}} \\
		& \leq \lim_{N \rightarrow \infty} \sum\limits_{i=1}^{N} \frac{1}{i^{2}} =  \frac{\pi^{2}}{6}< \infty.
	\end{align*}
	Thus, $\Phi(\omega^{N})(s)- \bar{P}^{N}(s) \rightarrow 0$ almost surely for all $s \in S$. SLLN implies that for $P$-almost every $\omega$ and for any $\epsilon > 0$, there exists $\bar{N}(\omega, \epsilon, s)$ such that for all $N \geq \bar{N}(\omega, \epsilon, s)$, $|\Phi(\omega^{N})(s)- \bar{P}^{N}(s)| < \epsilon $. By definition, for any $\epsilon > 0$, the data-revised set obtained using the empirical distribution method can be written as 
	\begin{equation*}
		\mathcal{P}(\omega^{N}) = \left\{P \in \mathcal{P}: \cap_{s \in S} |\bar{P}^{N}(s) -  \Phi(\omega^{N})(s)| < \epsilon \right\}.
	\end{equation*}
	Let $\bar{N}(\omega, \epsilon) = \max\limits \{\bar{N}(\omega,\epsilon,s)\}$, then SLLN implies that for all $N \geq \bar{N}(\omega, \epsilon)$, $P \in \mathcal{P}(\omega^{N})$. Thus, the conclusion holds for all $\epsilon > 0$.
\end{proof}
		
	\subsection{Proof of Theorem \ref{thm_confident}} 
	\begin{proof}[Proof of Theorem \ref{thm_confident}]
		Some additional notations are required to state this proof. Let $S = \{s_{1}, \cdots, s_{d}\}$. For $s \in \{s_{1}, \cdots, s_{d-1}\}$, let $e_{s}$ denote the corresponding standard basis in $\R^{d-1}$. Let $\mathbf{X_{i}}: \Omega \rightarrow \R^{d-1}$ denote a $(d-1)$-dimensional random vector such that  $\mathbf{X_{i}} = e_{\omega_{i}}$ if $\omega_{i} \in \{ s_{1},\cdots, s_{d-1} \}$ and $\mathbf{X_{i}} = 0$ if $\omega_{i} = s_{d}$. For each $P$, let $\mathbf{P_{i}}$ denote the $(d-1)$-dimensional vector whose coordinates are $P_{i}(s)$ for $s \in \{s_{1}, \cdots, s_{d-1}\}$. If $P$ is the underlying probability measure, then the mean vector of $\mathbf{X_{i}}$ is $\boldsymbol\mu_{i} = \mathbf{P_{i}}$. Moreover, the covariance matrix $\boldsymbol\Sigma_{\mathbf{i}}$ of the random vector $\mathbf{X_{i}}$ is 
		\begin{equation*}
			[\boldsymbol\Sigma_{\mathbf{i}}]_{kl} = 
			\begin{cases}
				P_{i}(s_{k})(1-P_{i}(s_{k})) \quad &\text{ if } k=l;\\
				-P_{i}(s_{k})P_{i}(s_{l}) \quad &\text{ if } k \neq l.
			\end{cases}
		\end{equation*}
		Given the full-support assumption, the covariance matrix is always positive definite. Let $\mathbf{\bar{X}_{N}} \equiv N^{-1}\sum_{i=1}^{N} \mathbf{X_{i}}$ be the random vector given by the average. Let $\boldsymbol{\bar{\mu}_{N}} \equiv N^{-1}\sum_{i=1}^{N} \boldsymbol{\mu_{\mathbf{i}}}$ and $\boldsymbol{\bar{\Sigma}_{N}} \equiv N^{-1} \sum_{i=1}^{N}\boldsymbol\Sigma_{\mathbf{i}}$. 
		
		\begin{lemma}\label{lem1}
			Let $\mathcal{P}$ be any compact subset of $\Delta_{indep}(\Omega)$ such that all $P \in \mathcal{P}$ have full support. For any $P \in \mathcal{P}$, the Central Limit Theorem holds, i.e., $\sqrt{N}(\mathbf{\bar{X}_{N{}}} - \boldsymbol{\bar{\mu}_{N}}) \xrightarrow{D} \mathcal{N}\left(0, \boldsymbol{\bar{\Sigma}_{N{}}} \right)$. 
		\end{lemma}
		
		The proof is standard by combining Liapounov's CLT for Triangular arrays with the Cramer-Wold device. Lemma \ref{lem1} implies one can use Gaussian approximations to construct acceptance regions. Specifically, for a $(d-1)$-dimensional random vector $\mathbf{X} \sim \mathcal{N}_{d-1}(\boldsymbol{\mu}, \boldsymbol\Sigma)$, the ellipsoidal region given by the set of all vectors $\mathbf{x}$ satisfying the following has a probability of $1-\alpha$: 
		\begin{equation*}
			(\mathbf{x} - \boldsymbol\mu)^{\intercal}\boldsymbol\Sigma^{-1} (\mathbf{x} - \boldsymbol\mu) \leq \chi_{d-1}^{2}(\alpha),
		\end{equation*}
		where $\chi_{d-1}^{2}(\alpha)$ is the upper $100\alpha\%$ quantile for the chi-square distribution with $d-1$ degrees of freedom. By definition, the set of all such vectors is the probability contour of the multivariate Gaussian distribution $\mathcal{N}_{d-1}(\boldsymbol{\mu}, \boldsymbol\Sigma)$ with probability $1-\alpha$. For any $P \in \Delta_{indep}(\Omega)$, define 
		\begin{equation}\label{equ_confidence_region}
			A^{*}_{N, \alpha}(P) = \{ \mathbf{x} \in \R^{d-1} : (\mathbf{x} - \boldsymbol{\bar{\mu}_{N}})^{\intercal}\boldsymbol{\bar{\Sigma}_{N}}^{-1} (\mathbf{x} - \boldsymbol{\bar{\mu}_{N}})  \leq \chi_{d-1}^{2}(\alpha) \}, 
		\end{equation}
		i.e., the probability contour of the multivariate Gaussian distribution $\mathcal{N}_{d-1}(\boldsymbol{\bar{\mu}_{N}}, \bar{\boldsymbol\Sigma}_{N})$ with probability $1-\alpha$. Then Lemma \ref{lem1} guarantees that 
		\begin{equation*}
			\liminf_{N \rightarrow \infty} P(A^{*}_{N,\alpha}(P)) \geq 1-\alpha.
		\end{equation*}
		Let $\boldsymbol{\hat{\Sigma}_{N}}$ denote the covariance matrix of the i.i.d. distribution $(\bar{P}^{N})^{\infty}$ and define
		\begin{equation}\label{equ_iid_confidence_region}
			A^{*}_{N,\alpha}((\bar{P}^{N})^{\infty}) = \{ \mathbf{x} \in \R^{d-1} : (\mathbf{x} - \boldsymbol{\bar{\mu}_{N}})^{\intercal}\boldsymbol{\hat{\Sigma}_{N}}^{-1} (\mathbf{x} - \boldsymbol{\bar{\mu}_{N}})  \leq \chi_{d-1}^{2}(\alpha) \}, 
		\end{equation}
		i.e., the corresponding probability contour for the multivariate Gaussian distribution $\mathbf{N}_{d-1}(\boldsymbol{\bar{\mu}_{N}}, \hat{\boldsymbol\Sigma}_{N})$. 
		
		\begin{lemma}\label{lem2}
			For any $P \in \Delta_{indep}(\Omega)$, $N$ and $\alpha$, $A^{*}_{N,\alpha}(P) \subseteq A^{*}_{N,\alpha}((\bar{P}^{N})^{\infty})$. 
		\end{lemma}
	
		Lemma \ref{lem2} further implies $P( A^{*}_{N,\alpha}(P)) \leq P(A^{*}_{N,\alpha}((\bar{P}^{N})^{\infty}))$. Therefore, it follows that 
		\begin{equation*}
			\liminf_{N \rightarrow \infty}  P(A^{*}_{N,\alpha}((\bar{P}^{N})^{\infty})) \geq \liminf_{N \rightarrow \infty} P(A^{*}_{N,\alpha}(P)) \geq 1-\alpha.
		\end{equation*}
		Finally, notice that 
		\begin{align*}
			\liminf_{N \rightarrow \infty}P(\{\omega^{N}: P_{N}^{\infty} \in \overline{co}(\mathcal{P}(\omega^{N})_{N}^{\infty}) \}) & \geq \liminf_{N \rightarrow \infty} P(\{ \omega^{N}: P \in \mathcal{P}(\omega^{N}) \}) \\
			&= \liminf_{N \rightarrow \infty} P(A^{*}_{N,\alpha}((\bar{P}^{N})^{\infty})) \geq 1-\alpha.
		\end{align*}
		Thus, the conclusion. 
	\end{proof}

\begin{proof}[Proof of Lemma \ref{lem1}]
	The following definitions are taken from \cite{WHITE1984107}:
	
	\textbf{Liapounov's CLT for Triangular Arrays.} 
	Let $\{ Z_{Ni} \}$ be a sequence of independent random scalars with $\mu_{Ni} \equiv E[Z_{Ni}]$, $\sigma_{Ni}^{2} = var(Z_{Ni})$, and $E|Z_{Ni} - \mu_{Ni}|^{2+\delta} < \Delta < \infty$ for some $\delta > 0$ and all $N$ and $i$. Define $\bar{Z}_{N} \equiv N^{-1}\sum_{i=1}^{N}Z_{Ni}$, $\bar{\mu}_{N} \equiv N^{-1}\sum_{i=1}^{N}\mu_{Ni}$ and $\bar{\sigma}_{N}^{2} \equiv var(\sqrt{N} \bar{Z}_{N}) = N^{-1}\sum_{i=1}^{N} \sigma_{Ni}^{2}$. If $\bar{\sigma}_{N}^{2} > \delta' > 0 $ for all $N$ sufficiently large, then 
	\begin{equation*}
		\sqrt{N}(\bar{Z}_{N} - \bar{\mu}_{N})/\bar{\sigma}_{N} \xrightarrow{D} \mathcal{N}(0,1).
	\end{equation*}
	
	\textbf{Cramer-Wold device.} Let $\{ \mathbf{b}_{N} \}$ be a sequence of random $k \times 1$  vectors and suppose that for any real $k \times 1$ vector $\boldsymbol\lambda$ such that $\boldsymbol\lambda^{\intercal}\boldsymbol\lambda = 1$, $\boldsymbol\lambda^{\intercal} \mathbf{b}_{N} \xrightarrow{D}\boldsymbol\lambda^{\intercal} \mathbf{Z}$ where $\mathbf{Z}$ is a $k\times 1$ vector with joint distribution function $F$. Then the limiting distribution of $ \mathbf{b}_{N}$ exists and equals $F$. \\
	
	For any $\boldsymbol\lambda$ with $\boldsymbol\lambda^{\intercal}\boldsymbol\lambda = 1$, consider the sequence of random variables given by $Y_{i} =  \boldsymbol\lambda^{\intercal}\mathbf{X_{i}}$. Then one has $E[Y_{i}] =\boldsymbol\lambda^{\intercal}\boldsymbol{\mu_{\mathbf{i}}}$ and $var(Y_{i}) = \boldsymbol\lambda^{\intercal} \boldsymbol\Sigma_{\mathbf{i}}\boldsymbol\lambda$. For $\bar{Y}_{N} = N^{-1}\sum_{i=1}^{N} Y_{i}$, one has $E[\bar{Y}_{N}] = \boldsymbol\lambda^{\intercal} \boldsymbol{\bar{\mu}_{N}}$ and
	\begin{equation*}
		var(\sqrt{N}\bar{Y}_{N}) =  \boldsymbol\lambda^{\intercal} \boldsymbol{\bar{\Sigma}_{N}}\boldsymbol\lambda = \boldsymbol\lambda^{\intercal}(\boldsymbol{\bar{\Sigma}_{N}}^{1/2})^{\intercal}(\boldsymbol{\bar{\Sigma}_{N}}^{1/2})\boldsymbol\lambda.
	\end{equation*}
	
	The existence of $\boldsymbol{\bar{\Sigma}_{N}}^{1/2}$ is guaranteed by the fact that $\boldsymbol{\bar{\Sigma}_{N}}$ is positive definite. Moreover, let it be the positive square root of $\boldsymbol{\bar{\Sigma}_{N}}$ so that it is unique, symmetric, and positive definite (Theorem 7.2.6 (a) in \citet{horn2012matrix}). Let $Z_{Ni} =\boldsymbol\lambda^{\intercal}\boldsymbol{\bar{\Sigma}_{N}}^{-1/2}\mathbf{X_{i}}$. Then one has $E[\bar{Z}_{N}] = \boldsymbol\lambda^{\intercal}\boldsymbol{\bar{\Sigma}_{N}}^{-1/2}\boldsymbol{\bar{\mu}}$ and 
	\begin{align*}
		var(\sqrt{N}\bar{Z}_{N}) & =  \boldsymbol\lambda^{\intercal} (\boldsymbol{\bar{\Sigma}_{N}}^{-1/2})^{\intercal} var(\sqrt{N}\mathbf{\bar{X}}_{N})\boldsymbol{\bar{\Sigma}_{N}}^{-1/2}\boldsymbol\lambda  \\
		& =  \boldsymbol\lambda^{\intercal} (\boldsymbol{\bar{\Sigma}_{N}}^{-1/2})^{\intercal} (\boldsymbol{\bar{\Sigma}_{N}}^{1/2})^{\intercal}(\boldsymbol{\bar{\Sigma}_{N}}^{1/2})(\boldsymbol{\bar{\Sigma}_{N}}^{-1/2})\boldsymbol\lambda= 1.
	\end{align*}
	If Liapounov's CLT holds, then  $\sqrt{N}(\bar{Z}_{N} - E[\bar{Z}_{N}])\xrightarrow{D} \mathcal{N}(0,1)$. Applying the Cramer-Wold device will imply the Lemma. Thus, it remains to show that the Liapounov condition holds for any $P \in \mathcal{P}$ and for any $\boldsymbol\lambda$. That is, for some $\delta > 0$, there exists $\Delta < \infty$ such that for any $N$ and $i$, $	E|Z_{Ni}|^{2 + \delta} < \Delta$. By Minkowski's inequality, 
	\begin{equation*}
		E[|Z_{Ni}|^{2 + \delta}] \leq \sum_{j}\lambda_{j} E|(\boldsymbol{\bar{\Sigma}_{N}}^{-1/2}\mathbf{X_{i}})_{j}|^{2+\delta}.
	\end{equation*}
	Notice that $\mathbf{X}_{i}$ is either $0$ or $e_{k}$. The RHS is less than the following
	\begin{equation*}
		(\max_{k,l} [\boldsymbol{\bar{\Sigma}_{N}}^{-1/2}]_{kl})^{2+\delta}.
	\end{equation*}
	Since $\boldsymbol{\bar{\Sigma}_{N}}^{-1/2}$ is symmetric and positive semi-definite, the largest entry must be on the diagonal. The trace of the matrix is equal to the sums of the eigenvalues. Both the diagonal elements and eigenvalues are non-negative due to positive definiteness (Corollary 7.1.5 in \cite{horn2012matrix}). Thus, it suffices to show the eigenvalues of $\boldsymbol{\bar{\Sigma}_{N}}^{-1/2}$ are bounded from above. Moreover, as each eigenvalue of $\boldsymbol{\bar{\Sigma}_{N}}^{-1/2}$ is the square root of the eigenvalue of $\boldsymbol{\bar{\Sigma}_{N}}^{-1}$. Thus, one only needs to show the eigenvalues of $\boldsymbol{\bar{\Sigma}_{N}}^{-1}$ are bounded from above, which is equivalent to showing the eigenvalues of $\boldsymbol{\bar{\Sigma}_{N}}$ are bounded away from zero. 
	
	For every $\boldsymbol\Sigma_{\mathbf{i}}$, the smallest eigenvalue is bounded below by $\min_{j}P_{i}(s_{j})$ \citep{watson1996}. Furthermore, as each $\boldsymbol\Sigma_{\mathbf{i}}$ is symmetric and positive definite, the minimum eigenvalue of $\sum_{i=1}^{n}\boldsymbol\Sigma_{\mathbf{i}}$ ($n$ times the minimum eigenvalue of $\boldsymbol{\bar{\Sigma}_{N}}$) is greater than the sum of the minimum eigenvalues of every $\boldsymbol\Sigma_{\mathbf{i}}$ (Corollary 4.3.15 in \cite{horn2012matrix}). Therefore, it suffices to show that $P_{i}(s_{j})$ is bounded away from zero for every $i$ and $j$. By assumption, $\mathcal{P}$ is compact thus $\min_{j,i} P_{i}(s_{j})$ exists. By full-support assumption, the minimum is always positive. 
\end{proof}

\begin{proof}[Proof of Lemma \ref{lem2}]
	It suffices to show that for any $P \in \Delta_{indep}(\Omega)$, $N$, constant $c \in \R$, and for all $\mathbf{x}$, 
	\begin{equation*}
		(\mathbf{x} - \boldsymbol{\bar{\mu}_{N}})^{\intercal}\boldsymbol{\bar{\Sigma}_{N}}^{-1} (\mathbf{x} - \boldsymbol{\bar{\mu}_{N}})\leq c \Rightarrow \boldsymbol (\mathbf{x} - \boldsymbol{\bar{\mu}_{N}})^{\intercal}\boldsymbol{\hat{\Sigma}_{N}} ^{-1} \boldsymbol (\mathbf{x} - \boldsymbol{\bar{\mu}_{N}}) \leq c,
	\end{equation*}
	which is equivalent to
	\begin{equation*}
		(\mathbf{x} - \boldsymbol{\bar{\mu}_{N}})^{\intercal}\boldsymbol{\bar{\Sigma}_{N}}^{-1} (\mathbf{x} - \boldsymbol{\bar{\mu}_{N}})-\boldsymbol (\mathbf{x} - \boldsymbol{\bar{\mu}_{N}})^{\intercal}\boldsymbol{\hat{\Sigma}_{N}}^{-1} \boldsymbol (\mathbf{x} - \boldsymbol{\bar{\mu}_{N}}) \geq 0 , \forall \mathbf{x}. 
	\end{equation*}
	In other words, for any $\mathbf{x}$, 
	\begin{equation*}
		\mathbf{x}^{\intercal} (\boldsymbol{\bar{\Sigma}_{N}}^{-1}-\boldsymbol{\hat{\Sigma}_{N}}^{-1})  \mathbf{x}\geq 0
	\end{equation*}
	That is, the lemma is true if $(\boldsymbol{\bar{\Sigma}_{N}}^{-1}-\boldsymbol{\hat{\Sigma}_{N}}^{-1})$ is a positive semi-definite matrix. It is further equivalent to, according to Corollary 7.7.4 in \cite{horn2012matrix}, $(\boldsymbol{\hat{\Sigma}_{N}} - \boldsymbol{\bar{\Sigma}_{N}})$ being positive semi-definite. The two covariance matrices are given by, respectively, 
	\begin{equation*}
		[\boldsymbol{\hat{\Sigma}}_{N}]_{kl} = 
		\begin{cases}
			\bar{P}^{N}(s_{k})(1-\bar{P}^{N}(s_{k})) \quad &\text{ if } k=l, \\
			-\bar{P}^{N}(s_{k})\bar{P}^{N}(s_{l}) \quad &\text{ if } k \neq l, 
		\end{cases}
	\end{equation*}
	and 
	\begin{equation*}
		[\boldsymbol{\bar{\Sigma}}_{N}]_{kl} = 
		\begin{cases}
			N^{-1}\sum\limits_{i=1}^{N} P_{i}(s_{k})(1-P_{i}(s_{k})) \quad &\text{ if } k=l,\\
			-N^{-1}\sum\limits_{i=1}^{N} P_{i}(s_{k})P_{i}(s_{l}) \quad &\text{ if } k \neq l.
		\end{cases}
	\end{equation*}
	By algebra, one can show, 
	\begin{equation*}
		N[\boldsymbol{\hat{\Sigma}}_{N} - \boldsymbol{\bar{\Sigma}}_{N}]_{kl} = 
		\begin{cases}
			\sum\limits_{i=1}^{N} (P_{i}(s_{k}) - \bar{P}^{N}(s_{k}))^{2}\quad &\text{ if } k=l;\\
			\sum\limits_{i=1}^{N} (P_{i}(s_{k}) - \bar{P}^{N}(s_{k}))(P_{i}(s_{l}) - \bar{P}^{N}(s_{l})) \quad &\text{ if } k \neq l.
		\end{cases}
	\end{equation*}
	
	Notice that $N[\boldsymbol{\hat{\Sigma}}_{N} - \boldsymbol{\bar{\Sigma}}_{N}]_{kl}$ is a Gram Matrix $G_{kl} = \langle v_{k}, v_{l}\rangle$, where the set of vectors are given by: 
	\begin{equation*}
		[v_{k}]_{i} = P_{i}(s_{k}) - \bar{P}^{N}(s_{k}).
	\end{equation*}
	A Gram Matrix is always positive semi-definite (Theorem 7.2.10 in \cite{horn2012matrix}), therefore $\boldsymbol{\hat{\Sigma}}_{N} - \boldsymbol{\bar{\Sigma}}_{N}$ is positive semi-definite as desired. 
\end{proof}

\section{When Experiments can be Dependent}\label{apx: dependent}
Throughout the main text, the possible data-generating processes are assumed to be independent across experiments. In this section, I show that all results in Section \ref{characterization} continue to hold after a slight modification of the definitions to take dependence into account. I then discuss how results in Section \ref{revision} can also be generalized accordingly.

For this section only, let the initial set $\mathcal{P}$ be a compact subset of $\Delta(\Omega)$. 
As before, assume that every $P\in\mathcal{P}$ has full support. 
For any $P\in\mathcal{P}$ and sample data $\omega^{N}\in S^{N}$, define the marginal distribution over future states conditional on $\omega^{N}$ by
\begin{equation}\label{equ_conditional}
    P(\tilde{\omega}_{N}\mid \omega^{N})
    = \frac{P(\omega^{N},\tilde{\omega}_{N})}{P(\omega^{N})},
    \qquad \forall\, \tilde{\omega}_{N}\in\Omega_{N}.
\end{equation}
Let $\mathcal{P}(\cdot\mid\omega^{N})$ denote the set of such conditional distributions.

The \emph{benchmark decision} continues to use the initial set $\mathcal{P}$ but conditions on the observed sample $\omega^{N}$ through the set of marginals $\mathcal{P}(\cdot\mid\omega^{N})$. 
This corresponds to the decision of a DM who applies prior-by-prior, or full-Bayesian, updating to each $P\in\mathcal{P}$. 
Formally,
\begin{equation*}
    c(D\mid\omega^{N})
    \equiv 
    \arg\max_{f\in D}
    \min_{P\in\mathcal{P}(\cdot\mid\omega^{N})}
        \int_{\Omega_{N}} f(\tilde{\omega}_{N})\, dP(\tilde{\omega}_{N})
    =
    \arg\max_{f\in D}
    \min_{P\in\overline{co}(\mathcal{P}(\cdot\mid\omega^{N}))}
        \int_{\Omega_{N}} f(\tilde{\omega}_{N})\, dP(\tilde{\omega}_{N}).
\end{equation*}

Let $\mathcal{P}(\omega^{N})$ denote the data-revised set, and $\mathcal{P}(\omega^{N})(\cdot\mid\omega^{N})$ the corresponding set of conditional distributions. 
The \emph{data-revised decision} is then
\begin{equation*}
    c(D,\omega^{N}\mid\omega^{N})
    \equiv 
    \arg\max_{f\in D}
    \min_{P\in\overline{co}(\mathcal{P}(\omega^{N})(\cdot\mid\omega^{N}))}
        \int_{\Omega_{N}} f(\tilde{\omega}_{N})\, dP(\tilde{\omega}_{N}).
\end{equation*}
In words, the data-revised decision is obtained through a two-step procedure: first revising the initial set to $\mathcal{P}(\omega^{N})$, and then conditioning each DGP in the revised set on the observed data.

The data-revised set \emph{accommodates} a DGP $P$ if
\begin{equation*}
    P(\cdot\mid\omega^{N})
    \in
    \overline{co}(\mathcal{P}(\omega^{N})(\cdot\mid\omega^{N})).
\end{equation*}
This condition holds whenever $P\in\overline{co}(\mathcal{P}(\omega^{N}))$, that is, whenever the data-revised set contains the true DGP. 
The data-revised set \emph{refines} the initial set if
\begin{equation*}
    \overline{co}(\mathcal{P}(\omega^{N})(\cdot\mid\omega^{N}))
    \subsetneqq
    \overline{co}(\mathcal{P}(\cdot\mid\omega^{N})).
\end{equation*}
A \emph{truth-accommodating refinement} is a data-revised set that both accommodates the true DGP $P^{*}$ and refines the initial set.
\bigskip 

Notice that all results in Section \ref{characterization} rely on the independence assumption only through the definitions of the benchmark decision, the data-revised decision, and the truth-accommodating refinement. Once these definitions are modified to incorporate dependence, all proofs continue to hold without change. More specifically, one can simply let $\mathcal{P}(\cdot\mid\omega^{N})$ be the initial set, $\mathcal{P}(\omega^{N})(\cdot\mid\omega^{N})$ the data-revised set, and $P^{*}(\cdot\mid\omega^{N})$ the true DGP, and directly apply the results.

\begin{theorem}
Suppose $\mathcal{P} \subset \Delta(\Omega)$. If the benchmark decision, data-revised decision, and truth-accommodating refinement are defined as in this section, then all results in Section \ref{characterization} continue to hold. 
\end{theorem}

For results in Section \ref{revision}, as already noted in the main text, Theorem \ref{thm_asymptotic} continues to hold as long as a version of the strong law of large numbers applies. Theorem \ref{thm_confident} is more delicate, but by the same token, it is natural to conjecture that it can also be generalized under similar conditions. 

Finally, Theorem \ref{thm_asymptotic} can be further extended to the case where $S$ is infinite. In that case, the proof would invoke the generalized Glivenko-Cantelli theorem for independent but non-identically distributed random variables, as presented in \cite{WELLNER1981309}.

\bibliographystyle{econ}
\bibliography{mylibrary.bib}

\end{document}